\documentclass{llncs}

\usepackage[utf8]{inputenc}
\usepackage[T1]{fontenc}

\newcommand{\myspacea}{}
\newcommand{\myspaceb}{}
\newcommand{\myspaceab}{}

\usepackage{amssymb,amsmath,amsfonts}
\usepackage{tikz}
\usetikzlibrary{arrows,shapes,snakes,automata,backgrounds,petri,calc,fit,positioning,automata,patterns}

\usepackage{xspace}
\usepackage{txfonts}
\usepackage{algorithmic}
\setboolean{ALC@noend}{true} 
\usepackage{algorithm}
\usepackage{fancyvrb}
\usepackage{enumerate}
\usepackage{wrapfig}

\tikzstyle{max}=[thick,draw,minimum size=1.4em,inner sep=0em]
\tikzstyle{min}=[diamond,thick,draw,minimum size=1.4em,%
    inner sep=0em]
\tikzstyle{ran}=[circle,thick,draw,minimum size=1.4em,%
    inner sep=0em]
\tikzstyle{act}=[circle,thick,draw,fill,minimum size=.7em,%
    inner sep=0em]
\tikzstyle{mc}=[rounded corners,thick,draw,minimum size=1.4em,%
    inner sep=.5ex]
\tikzstyle{tran}=[thick,draw,->,>=stealth]
\tikzstyle{loop left}=[tran, to path={.. controls +(150:.5)
    and +(210:.5) .. (\tikztotarget) \tikztonodes}]
\tikzstyle{loop right}=[tran, to path={.. controls +(30:.5)
    and +(330:.5) .. (\tikztotarget) \tikztonodes}]
\tikzstyle{loop above}=[tran, to path={.. controls +(60:.5)
    and +(120:.5) .. (\tikztotarget) \tikztonodes}]
\tikzstyle{loop below}=[tran, to path={.. controls +(240:.5)
    and +(300:.5) .. (\tikztotarget) \tikztonodes}]

\newcommand{\theoremlike}[2]{\par\medskip\penalty-250\refstepcounter{theorem}{{\bfseries\noindent#2
\ref{#1}.}}}

\newcommand{\thmhelperpre}[2]{\theoremlike{#1}{#2}}
\newcommand{\thmhelperpost}{\par\medskip}

\newenvironment{reflemma}[1]{\thmhelperpre{#1}{Lemma}}{\thmhelperpost}

\newenvironment{refproposition}[1]{\thmhelperpre{#1}{Proposition}}{
\thmhelperpost }

\newcommand{\QED}{\qed}

\newcommand{\Nset}{\mathbb{N}}
\newcommand{\Nseto}{\Nset_0}

\newcommand{\Rset}{\mathbb{R}}
\newcommand{\Rsetp}{\mathbb{R}_{>0}}
\newcommand{\Rsetpo}{\mathbb{R}_{\ge 0}}

\renewcommand{\vec}[1]{\mathbf{#1}}

\newcommand{\indicator}[1]{\mathbf{1}[#1]}

\newcommand{\probm}{\mathcal{P}}

\newcommand{\Productfield}{\bigotimes}

\newcommand{\sigmafield}{\mathcal{F}}
\newcommand{\initmeasure}{\mu}

\newcommand{\de}[1]{\mathit{d#1}} 

\newcommand{\dist}{\mathcal{D}}
\newcommand{\distribution}{\alpha}

\newcommand{\events}{\mathcal{E}}
\newcommand{\fixed}{\events_f}

\newcommand{\nice}{single-ticking\xspace}
\newcommand{\emnice}{\emph{single-ticking}\xspace}

\newcommand{\states}{S}
\newcommand{\sched}{\mathbf{E}}
\newcommand{\occur}{\mathrm{Succ}}

\newcommand{\val}{\nu} 
\newcommand{\bottom}{\bot} 
\newcommand{\last}{\vartriangle}

\newcommand{\configs}{\Gamma}
\newcommand{\configsfield}{\mathcal{G}}

\newcommand{\kernel}{P}

\newcommand{\hit}{\mathrm{Hit}}
\newcommand{\win}{\mathrm{Win}}

\newcommand{\gssmc}{\Phi}

\newcommand{\allruns}{\Omega}

\newcommand{\reach}{\mathit{Reach}}

\newcommand{\dme}{\mathbf{d}}

\newcommand{\cme}{\mathbf{c}}

\newcommand{\measured}{{\mathring{s}}}

\newcommand{\fr}{\mathrm{frac}}  
\newcommand{\intpart}{\mathrm{int}}
\newcommand{\region}{\sim} 
\newcommand{\regiongraph}{G} 
\newcommand{\sizereg}{|V|}

\newcommand{\smallmeasure}{\kappa}
\newcommand{\densb}{c}
\newcommand{\pmin}{p_{\min}}

\newcommand{\setofruns}{\mathcal{R}}

\newcommand{\restart}{Sink}

\newcommand{\todo}[1]{%
  { \renewcommand{\baselinestretch}{0.3} %
  \begin{tikzpicture}[remember picture]%
      \node [coordinate] (inText) {};%
  \end{tikzpicture} %
  \marginpar[{
      \begin{tikzpicture}[remember picture]%
          \draw node[draw=red, anchor=north west, color=red, text width = 4cm, font=\footnotesize, inner sep=0.10cm,yshift=-0.15cm] (inNote)%
                   {#1};%
      \end{tikzpicture}%
      \begin{tikzpicture}[remember picture, overlay] %
    \draw[draw = red]%
        ([yshift=-0.10cm] inText)%
        -| ([xshift=0.05cm] inNote.south east) %
        -| ([xshift=0.05cm] inNote.east) %
        -| (inNote.east);%
      \end{tikzpicture}%
  }]{ 
      \begin{tikzpicture}[remember picture]%
          \draw node[draw=red, anchor=north west, color=red, text width = 4cm, font=\footnotesize, inner sep=0.10cm,yshift=-0.15cm] (inNote)%
                   {#1};%
      \end{tikzpicture}%
      \begin{tikzpicture}[remember picture, overlay] %
    \draw[draw = red]%
        ([yshift=-0.10cm] inText)%
        -| ([xshift=-0.05cm] inNote.south west)%
        -| ([xshift=-0.05cm] inNote.west)%
        -| (inNote.west);%
      \end{tikzpicture}%
  }}}%
\renewcommand{\todo}[1]{}

\newcommand{\appref}[1]{ Appendix~\ref{#1}}

\begin{document}

\title{Fixed-delay Events in Generalized\\ Semi-Markov Processes Revisited\thanks{
The authors are supported by the Institute for Theoretical Computer Science, project No.~1M0545,  the Czech Science Foundation, grant No.~P202/10/1469 (T.~Br\'{a}zdil, V.~\v{R}eh\'ak) and No.~102/09/H042 (J.~Kr\v{c}\'al), and Brno PhD Talent Financial Aid (J.~K\v ret\'insk\'y).
}}
\toctitle{Fixed-delay Events in Generalized Semi-Markov Processes Revisited}
\titlerunning{Fixed-delay Events in GSMP Revisited}

\author{Tom\'a\v{s} Br\'azdil \and Jan Kr\v{c}\'al \and Jan
  K{\v{r}}et\'{\i}nsk\'y\thanks{
On leave at TU M\"{u}nchen, Boltzmannstr.{} 3, Garching, Germany.} 
\and Vojt\v{e}ch {\v{R}}eh\'ak}

\institute{Faculty of Informatics, Masaryk University,
  Brno, Czech Republic\\
  {\{brazdil,\,krcal,\,jan.kretinsky,\,rehak\}@fi.muni.cz}
  }

\maketitle

\begin{abstract}
  We study long run average behavior of generalized semi-Markov processes
  with both fixed-delay events as well as variable-delay events.
  We show that allowing two fixed-delay events and one variable-delay event
  may cause an unstable behavior of a GSMP. In particular, we show that a
  frequency of a given state may not be defined for almost all runs (or more
  generally, an invariant measure may not exist). We use this observation to
  disprove several results from literature. Next we study GSMP with at most
  one fixed-delay event combined with an arbitrary number of variable-delay
  events. We prove that such a GSMP always possesses an invariant measure
  which means that the frequencies of states are always well defined and we
  provide algorithms for approximation of these frequencies. Additionally,
  we show that the positive results remain valid even if we allow an
  arbitrary number of reasonably restricted fixed-delay events.
\end{abstract}

\myspacea\myspacea

\section{Introduction}
\label{sec-intro}

\myspacea

Generalized semi-Markov processes (GSMP), introduced by Matthes in \cite{Matthes:GSMP}, are a standard model for discrete-event stochastic systems. Such a system operates in continuous time and reacts, by changing its state, to occurrences of events. 
Each event is assigned a random delay after which it occurs; state transitions may be randomized as well.
Whenever the system reacts to an event, new events may be scheduled and pending events may be discarded. 
To get some intuition, imagine a simple communication model in which a server sends messages to several clients asking them to reply. The reaction of each client may be randomly delayed, e.g., due to latency of communication links. Whenever a reply comes from a client, the server changes its state (e.g., by updating its database of alive clients or by sending another message to the client) and then waits for the rest of the replies. Such a model is usually extended by allowing the server to time-out and to take an appropriate action, e.g., demand replies from the remaining clients in a more urgent way. The time-out can be seen as another event which has a fixed delay.

%

More formally, a GSMP consists of a set $\states$ of states and a set $\events$ of events. 
 Each state $s$ is
assigned a set $\sched(s)$ of events {\it scheduled} in $s$. Intuitively, each
event in $\sched(s)$ is assigned a positive real number representing the
amount of time which elapses before the event occurs. Note that several events may occur
at the same time.
Once a set of events $E\subseteq \sched(s)$ occurs, the system makes a {\it transition} to a new
state $s'$. The state $s'$ is randomly chosen according to a fixed distribution which depends only on the
state $s$ and the set $E$. 
In $s'$, the \emph{old} 
events of $\sched(s)\smallsetminus \sched(s')$ are
discarded, each \emph{inherited} event of $(\sched(s')\cap
\sched(s))\smallsetminus E$ remains scheduled to the same point in the future, 
and each \emph{new} event of $(\sched(s')\smallsetminus \sched(s))\cup (\sched(s')\cap E)$
is newly scheduled
according to its given probability distribution.

In order to deal with GSMP in a rigorous way, one has to impose some
restrictions on the distributions of delays. Standard mathematical literature, such
as~\cite{Haas99,Haas:book}, usually considers GSMP with continuously distributed delays.
This is certainly a limitation, as some systems with fixed time delays
(such as time-outs or processor ticks)
cannot be faithfully modeled using only continuously distributed delays. 
We show some examples where fixed delays exhibit qualitatively different
behavior than any continuously distributed approximation.
In this paper we consider the following two types of events:
\begin{itemize}
\item \emph{variable-delay}: the delay of the event is randomly distributed according to
  a probability density function which is continuous and positive either on a bounded interval $[\ell,u]$ or on an unbounded interval $[\ell,\infty)$;
\item \emph{fixed-delay}: the delay is set to a fixed value with probability one.
\end{itemize}
The desired behavior of systems modeled using GSMP can be specified by
various means. One is often
interested in long-run behavior such as mean
response time, frequency of errors, etc. (see, e.g.,~\cite{Alfaro:Long-run-average}). For example, in the above communication model, one may be interested in average response time of clients or in average time in which all clients eventually reply.
Several model independent
formalisms have been devised for expressing such properties of continuous
time systems. For example, a well known temporal logic CSL contains a steady
state operator expressing frequency of states satisfying a given
subformula. In~\cite{BKKKR-HSSC2011}, we proposed to specify long-run behavior of a
continuous-time process using a timed automaton which observes runs of the
process, and measure the frequency of locations of the automaton.

In this paper we consider a standard performance measure, the frequency of
states of the GSMP. To be more specific, 
let us fix a state $\measured\in S$. We define a random variable $\dme$ which to
every run  assigns the (discrete) frequency of visits to $\measured$ on the run, i.e. the ratio of the number of transitions
entering $\measured$ to the number of all transitions.
We also define a random variable $\cme$ which gives timed frequency of $\measured$, i.e. the ratio of the amount of time spent in $\measured$ to the amount of time spent in all states. Technically, both variables $\dme$ and $\cme$ are defined as limits of the corresponding ratios on prefixes of
the run that are prolonged ad infinitum.
Note that the limits may not be defined for some runs. 
For example, consider a run which alternates between $\measured$ and another state $s$; it spends $2$ time unit in $\measured$, then $4$ in $s$, then $8$ in $\measured$, then $16$ in $s$, etc. Such a run does not have a limit ratio between time spent in $\measured$ and in $s$. We say that $\dme$ (or $\cme$) is well-defined for a run if the limit ratios exist for this run.
Our goal is to characterize stable systems that have the variables $\dme$ and $\cme$ well-defined 
for almost all runs, and to analyze the probability distributions of $\dme$
and $\cme$ on these stable systems.  


As a working example of GSMP with fixed-delay events, we present a simplified protocol for time
synchronization.  Using the variable $\cme$, we show how to measure
reliability of the protocol. 
Via message exchange, the protocol sets and keeps a
client clock sufficiently close to a server clock. 
%
%
Each message exchange is initialized by the client asking the server for the
current time, i.e.~sending a \emph{query} message. The server adds a
timestamp into the message and sends it back as a \emph{response}. This
query-response exchange provides a reliable data for \emph{synchronization}
action if it is realized within a given \emph{round-trip delay}.
%
%
Otherwise, the client has to repeat the procedure.
%
%
After a success, the client is considered to be synchronized until a given \emph{stable-time delay} elapses.
Since the aim is to keep the clocks synchronized all the time,
the client restarts the synchronization process sooner, i.e. after a given \emph{polling
  delay} that is shorter than the stable-time delay. Notice that the client gets
desynchronized whenever several unsuccessful synchronizations occur in a row. 
Our goal is to measure the portion of the time when the
client clock is not synchronized.

Figure~\ref{fig-GSMP-intro} shows a GSMP model of this protocol. 
The delays specified in
the protocol are modeled using fixed-delay events $roundtrip\_d$,
$stable\_d$, and $polling\_d$ while actions are modeled by variable-delay
events $query$, $response$, and $sync$.
%
%
%
Note that if the stable-time runs out before a fast enough response arrives,
the systems moves into primed states denoting it is not synchronized at the
moment. 
Thus,
$\cme(\textit{Init'})+\cme(\textit{Q-sent'})$ expresses the portion of the
time when the client clock is not synchronized.

\begin{figure}[t]\myspaceb
  \centering
     
\tikzstyle{max}=[thick,draw,minimum size=1.4em,inner sep=0em]
\tikzstyle{min}=[diamond,thick,draw,minimum size=2em,%
    inner sep=0em]
\tikzstyle{ran}=[circle,thick,draw,text width=14
mm,text centered,minimum size=2em,%
    inner sep=0em]
\tikzstyle{act}=[circle,thick,draw,fill,minimum size=.7em,%
    inner sep=0em]
\tikzstyle{mc}=[rounded corners,thick,draw,minimum size=1.4em,%
    inner sep=.5ex]
\tikzstyle{tran}=[thick,draw,->,>=stealth]
\tikzstyle{loop left}=[tran, to path={.. controls +(150:.7)
    and +(210:.7) .. (\tikztotarget) \tikztonodes}]
\tikzstyle{loop right}=[tran, to path={.. controls +(30:.7)
    and +(330:.7) .. (\tikztotarget) \tikztonodes}]
\tikzstyle{loop above}=[tran, to path={.. controls +(60:.7)
    and +(120:.7) .. (\tikztotarget) \tikztonodes}]
\tikzstyle{loop below}=[tran, to path={.. controls +(240:.7)
    and +(300:.7) .. (\tikztotarget) \tikztonodes}]

 \begin{tikzpicture}[x=1.7cm,y=1.05cm],font=\scriptsize]
\node (Init) at (0,0) [ran] { \textit{Init}\\ $\{query,$\\$stable\_d\}$};
\node (Qsent) at (2,0) [ran] {\textit{Q-sent}\\ $\{response,$\\$roundtrip\_d,$\\$stable\_d\}$};
\node (Rrecv) at (4,0) [ran] {\textit{R-recvd}\\ $\{sync\}$};
\node (Idle) at (6,0) [ran] {\textit{Idle}\\ $\{stable\_d,$\\$polling\_d\}$};

\node (Init') at (0,-2) [ran] { \textit{Init'}\\ $\{query\}$};
\node (Qsent') at (2,-2) [ran] {\textit{Q-sent'}\\ $\{response,$\\$roundtrip\_d\}$};

\draw [tran,rounded corners] (Init) 
  edge 
  node[above] {$query$} 
 (Qsent);

\draw [tran,rounded corners] (Qsent) 
  edge 
  node[above] {$response$} 
  (Rrecv);

\draw [tran,rounded corners] (Rrecv) 
  edge 
  node[above] {$sync$} 
  (Idle);

\draw [tran,rounded corners] (Qsent) 
  -- +(-0.4,0.8)
  -- +(-1.6,0.8) 
  node[below, pos=0.5] {$roundtrip\_d$} 
  -- (Init);

\draw [tran,rounded corners] (Init') 
  edge 
  node[above] {$query$} 
 (Qsent');

\draw [tran,rounded corners] (Qsent') 
  edge 
  node[near start,below right] {$response$} 
  (Rrecv);

\draw [tran,rounded corners] (Qsent') 
  -- +(-0.4,-0.8)
  -- +(-1.6,-0.8) 
  node[above, pos=0.5] {$roundtrip\_d$} 
  -- (Init');

\draw [tran,rounded corners] (Idle) 
  -- +(0,1.2)
  -- +(-6.0,1.2) 
  node[below, pos=0.1667] {$polling\_d$} 
  -- (Init);


\foreach \i in {Init, Qsent}  {
  \draw [tran,rounded corners] (\i) 
    edge 
    node[left] {$stable\_d$} 
    (\i');
}

\end{tikzpicture}
\caption{A GSMP model of a clock synchronization protocol. Below each
state label, we list the set of scheduled events. We only display transitions that can
take place with non-zero probability.}
\label{fig-GSMP-intro}\myspacea\myspacea
\end{figure}
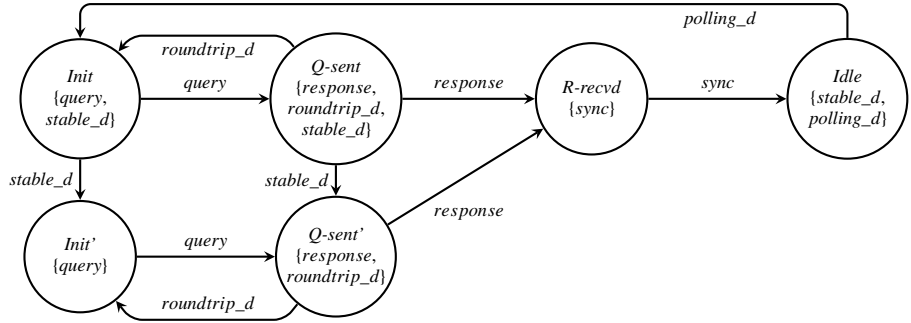
%

\myspacea

\subsubsection{Our contribution.}
So far, GSMP were mostly studied with variable-delay events only. There are
a few exceptions such as~\cite{ACD:verifying-automata-real-time,ACD:model-checking-real-time,BA:gsmp-one-stateful-clock,AB:gsmp-bounded-model-checking} but 
they often contain erroneous statements due to presence of fixed-delay events. Our goal is to study the effect of mixing  a number of fixed-delay events with an arbitrary amount of variable-delay events. 

At the beginning we give an example of a GSMP with two fixed-delay events
for which it is \emph {not true} that the variables $\dme$ and $\cme$ are
well-defined for almost all runs.
We also disprove some crucial statements
of~\cite{ACD:model-checking-real-time,ACD:verifying-automata-real-time}. In
particular, we show an example of a GSMP
which 
reaches one of its states with probability less than one even though the
algorithms
of~\cite{ACD:model-checking-real-time,ACD:verifying-automata-real-time} return 
the probability one.  The mistake of these algorithms is fundamental as they
neglect the possibility of unstable behavior of GSMP.
%

Concerning positive results, we
show that if there is at most one fixed-delay event, then both $\dme$ and $\cme$ are almost surely well-defined.
This is true even if we allow an arbitrary number of 
reasonably restricted
fixed-delay events. 
We also show how to 
approximate 
distribution functions of $\dme$ and $\cme$. 
To be more specific, we show that
for GSMP with at most one unrestricted and an arbitrary number of restricted fixed-delay events, both variables $\dme$ and $\cme$\ have finite ranges $\{d_1,\ldots,d_n\}$ and $\{c_1,\ldots,c_n\}$. Moreover, 
all values $d_i$ and $c_i$ 
and
probabilities $\probm(\dme=d_i)$ and $\probm(\cme=c_i)$ 
can be effectively approximated.


\myspacea

\subsubsection{Related work.}
There are two main approaches to the analysis of GSMP. 
One is to restrict the amount of events or types of their distributions and to solve the problems using symbolic methods~\cite{BA:gsmp-one-stateful-clock,AB:gsmp-bounded-model-checking,MC:dspn-defined}. The other is to estimate the values of interest using simulation~\cite{YS:gsmp-monte-carlo,Haas99,Haas:book}.
Concerning the first approach, time-bounded reachability has been studied in~\cite{AB:gsmp-bounded-model-checking} where the authors restricted the delays of events to so called expolynomial distributions.
%
The same authors also studied reachability probabilities of GSMP where in each transition at most one event is inherited~\cite{BA:gsmp-one-stateful-clock}. 
Further, the widely studied formalisms of semi-Markov processes (see, e.g., \cite{LHK:smp-csl-model-checking,BKKKR-HSSC2011}) and continuous-time Markov chains (see, e.g., \cite{BHHK:CTMC-model-checking-IEEE-TSE,BCHKM:ctmc-ta-efficient}) are both subclasses of GSMP. 
%
%

As for the second approach, GSMP are studied by mathematicians as a standard model 
for discrete event simulation and Markov chains Monte Carlo
(see, e.g., \cite{Glynn:GSMP,HS:GSMP-regenerative,RR:GSSMC-PS}).
Our work is strongly related to~\cite{Haas99,Haas:book} where the long-run average behavior
of GSMP with variable-delay events is studied. 
Under relatively standard
assumptions the stochastic process generated by a GSMP is shown to be irreducible and
to possess an invariant measure. In such a case, the variables $\dme$ and $\cme$ are almost
surely constant. 
Beside the theoretical results, there exist tools that employ simulation for model checking (see, e.g., \cite{YS:gsmp-monte-carlo,CJMS:smart-tool}).

In addition, GSMP are a proper subset of stochastic automata, a model of concurrent systems (see, e.g., \cite{DK:stochastic-automata}).
Further, as shown in~\cite{Haas:book}, GSMP have the same modeling power as stochastic Petri nets
~\cite{SPN-modeling:book}. The formalism of
deterministic and stochastic Petri nets (DSPN) introduced
by~\cite{MC:dspn-defined} 
adds deterministic transitions -- a counterpart of fixed-delay events.
The authors restricted the model to at most one deterministic transition enabled at a time
and to exponentially distributed timed transitions. For this restricted model, 
the authors proved existence of a steady state distribution and provided an algorithm for its computation. However, the methods inherently rely on the properties of the exponential distribution and cannot be extended to our setting with general variable delays.
DSPN have been extended by \cite{GL:dspn-analysis-supplementary-variables,LS:dspn-numerical-analysis} to allow arbitrarily many deterministic transitions. The authors provide algorithms for steady-state analysis of DSPN
that were implemented in the tool DSPNExpress~\cite{LRT:dspnexpress-tool},
but do not discuss under which conditions the steady-state distributions exist.

\myspacea\myspacea

\section{Preliminaries}

\myspacea

In this paper, the sets of all positive integers, non-negative
integers, real numbers, positive real numbers, and non-negative real
numbers are denoted by 
$\Nset$, $\Nseto$, $\Rset$, $\Rsetp$, and $\Rsetpo$, respectively. 
For a real number $r\in\Rset$, $\intpart(r)$ denotes its integral part, 
i.e. the largest integer smaller than $r$, and $\fr(r)$ denotes its fractional part, 
i.e. $r-\intpart(r)$.
%
Let $A$ be a finite or countably infinite set. A  
\emph{probability distribution}
on $A$ is a function $f : A \rightarrow \Rsetpo$ such that
\mbox{$\sum_{a \in A} f(a) = 1$}. 
The set of all distributions
on $A$ is denoted by $\dist(A)$.

A \emph{$\sigma$-field} over a set $\Omega$ is a set $\sigmafield \subseteq 2^{\Omega}$
that includes $\Omega$ and is closed under complement and countable union. 
A \emph{measurable space} is a pair $(\Omega,\sigmafield)$ where $\Omega$ is a
set called \emph{sample space} and $\sigmafield$ is a $\sigma$-field over $\Omega$
whose elements are called \emph{measurable sets}. 
Given a measurable space $(\Omega,\sigmafield)$, we say that a function $f : \Omega \to \Rset$ 
is a random variable if the inverse image of any real interval is a measurable set.
A \emph{probability measure} over a measurable space
$(\Omega,\sigmafield)$ is a function $\probm : \sigmafield \rightarrow \Rsetpo$
such that, for each countable collection $\{X_i\}_{i\in I}$ of pairwise
disjoint elements of $\sigmafield$, we have $\probm(\bigcup_{i\in I} X_i) = 
\sum_{i\in I} \probm(X_i)$ and, moreover, $\probm(\Omega)=1$. A 
\emph{probability space} is a triple $(\Omega,\sigmafield,\probm)$, where 
$(\Omega,\sigmafield)$ is a measurable space and $\probm$
is a probability measure over $(\Omega,\sigmafield)$. 
We say that a property $A \subseteq \Omega$ holds 
for \emph{almost all} elements
of a measurable set $Y$ if $\probm(Y) > 0$,
$A \cap Y \in \sigmafield$, and $\probm(A \cap Y \mid Y) = 1$.
Alternatively, we say that $A$ holds \emph{almost surely} for $Y$.


\myspacea

\subsection{Generalized semi-Markov processes}
\label{sec-GSMP}\myspacea
Let $\events$ be a finite set of \emph{events}. To every $e \in \events$ we associate the \emph{lower bound} 
$\ell_e \in \Nseto$ and the \emph{upper bound} $u_e \in \Nset \cup \{\infty\}$ of its delay.
We say that $e$ is a \emph{fixed-delay} event if $\ell_e = u_e$, and a \emph{variable-delay} event if $\ell_e < u_e$.
Furthermore, we say that a variable-delay event $e$ is \emph{bounded} if 
$u_e \neq \infty$, and \emph{unbounded}, otherwise.
To each variable-delay event $e$ we assign a \emph{density function} $f_e : \Rset \rightarrow \Rset$ such that
$\int_{\ell_e}^{u_e} f_e(x)\, \de{x} = 1$. We assume 
$f_e$ to be positive and continuous on the whole $[\ell_e,u_e]$ or $[\ell_e,\infty)$ if $e$ is 
bounded or unbounded, respectively, and zero elsewhere. 
We require that $f_e$ have finite expected value, i.e.~$\int_{\ell_e}^{u_e} x \cdot f_e(x)\, \de{x} < \infty$.
 
\begin{definition}
  A \emph{generalized semi-Markov process} is a tuple 
  $(\states,\events,\sched,\occur,\distribution_0)$ where 
\begin{itemize}
 \item $\states$ is a finite set of \emph{states}, 
 \item $\events$ is a finite set of \emph{events},
 \item $\sched : \states \rightarrow 2^{\events}$ assigns to each state $s$ a set of events $\sched(s) \neq \emptyset$
  \emph{scheduled} to occur in $s$,
 \item $\occur : \states \times 2^{\events} \rightarrow \dist(\states)$ is the \emph{successor} function, i.e. 
  assigns a probability distribution specifying the successor state
  to each state and set of events that occur simultaneously in this state, and
 \item $\distribution_0 \in \dist(\states)$ is the \emph{initial distribution}.
\end{itemize}
\end{definition}
A \emph{configuration} is a pair $(s,\val)$ where $s \in
S$ and $\val$ is a {\it valuation} which assigns to every event $e\in
\sched(s)$ the amount of time that elapsed since the event $e$ was
scheduled.%
\footnote{Usually, the valuation is defined to store the time left before the event appears. However, our definition is equivalent and more convenient for the general setting where both bounded and unbounded events appear.}
 For convenience, we define $\nu(e) = \bottom$ whenever $e\not\in
\sched(s)$, and we denote by $\nu(\last)$ the amount of time spent in the
previous configuration (initially, we put $\nu(\last)=0$). When a set of events $E$ occurs
and the process moves from $s$ to a state $s'$, the valuation of \emph{old} events of
$\sched(s)\smallsetminus \sched(s')$ is discarded to $\bottom$, the
valuation of each \emph{inherited} event of $(\sched(s')\cap
\sched(s))\smallsetminus E$ is increased by the time spent in $s$, and the
valuation of each \emph{new} event of $(\sched(s')\smallsetminus
\sched(s))\cup (\sched(s')\cap E)$ is set to $0$.

We illustrate the dynamics of GSMP on the example of
Figure~\ref{fig-GSMP-intro}. Let the bounds of the fixed-delay events
$roundtrip\_d$, $polling\_d$, and $stable\_d$ be $1$, $90$, and $100$,
respectively. 
We start in the state \textit{Idle}, i.e.~in
the configuration $(\textit{Idle},((polling\_d,0),(stable\_d,0),(\last,0)))$
denoting that $\nu(polling\_d)=0$, $\nu(stable\_d)=0$, $\nu(\last)=0$, and
$\bottom$ is assigned to all other events. After 90 time units, the event
$polling\_d$ occurs and we move to
$(\textit{Init},((query,0),(stable\_d,90),(\last,90)))$. Assume that the
event $query$ occurs in the state \textit{Init} after $0.6$ time units and
we move to 
$(\textit{Q-sent},((response,0),(roundtrip\_d,0),(stable\_d,90.6),(\last,0.6)))$
and so forth. 


A formal semantics of GSMP is usually defined in terms of general state-space Markov chains (GSSMC, see, e.g., \cite{MT:book}).
A GSSMC is a stochastic process $\Phi$ over a measurable state-space $(\configs, \configsfield)$
whose dynamics is determined by an initial measure $\initmeasure$ on $(\configs, \configsfield)$ and a \emph{transition kernel} $\kernel$ which specifies one-step transition probabilities.\footnote{
Precisely, transition kernel
is a function $\kernel: \configs \times \configsfield \to[0,1]$ such that
$\kernel(z,\cdot)$ is a probability measure over $(\configs,\configsfield)$ for
each $z\in\configs$; and $\kernel(\cdot,A)$ is a measurable function for each $A\in\configsfield$.}
A~given GSMP induces a GSSMC whose state-space consists of all configurations, 
the initial measure $\initmeasure$ is induced by $\distribution_0$ in a natural way, 
and the transition kernel is determined by the dynamics of GSMP described above. Formally,
\begin{itemize}
\item $\configs$ 
  is the set of all configurations, and $\configsfield$ is a $\sigma$-field
  over $\configs$ induced by the discrete topology over $S$ and the Borel
  $\sigma$-field over the set of all valuations;
\item the initial measure $\initmeasure$ allows to start in configurations
  with zero valuation only, i.e. for $A\in\configsfield$ we have
  $\initmeasure(A) = \sum_{s\in \mathit{Zero}(A)} \distribution_0(s)$ where
  $\mathit{Zero}(A) = \{ s \in S \mid (s,\vec{0}) \in A \}$;
\item the transition kernel $\kernel(z,A)$ describing the probability to move in 
one step from a configuration $z = (s,\nu)$ to any configuration in a set $A$ is defined as follows. 
It suffices to consider $A$ of the form $\{s'\} \times X$ where $X$ is a measurable
set of valuations. Let $V$ and $F$ be the sets of
  variable-delay and fixed-delay events, respectively, that are scheduled in $s$.
  Let $F' \subseteq F$ be the set of fixed-delay events that can occur as first among 
 the fixed-delay event enabled in $z$, i.e. that have in $\nu$ the minimal remaining time $u$.  
Note that two variable-delay events occur simultaneously with probability zero.
  Hence, we consider all combinations of $e\in V$ and $t\in\Rsetpo$ stating
  that
  \begin{equation*}
    P(z,A) = 
    \begin{cases}
      \sum_{e\in V} \int_0^\infty \hit(\{e\},t) \cdot \win(\{e\},t)\; \de{t}
      & \text{if $F = \emptyset$} \\
      \sum_{e\in V} \int_0^{u\phantom{tt}} \hit(\{e\},t) \cdot
      \win(\{e\},t)\; \de{t} + \hit(F',u)\cdot\win(F',u) & \text{otherwise,}
    \end{cases}
  \end{equation*}
  where the term $\hit(E,t)$ denotes the conditional probability of hitting
  $A$ under the condition that $E$ occurs at time $t$ and the term
  $\win(E,t)$ denotes the probability (density) of $E$ occurring at time $t$.
  Formally, $$\hit(E,t) = \occur(s,E)(s') \cdot \indicator{\nu' \in X} $$
  where $\indicator{\nu' \in X}$ is the indicator function and 
  $\nu'$ is the valuation after the transition, i.e.~$\nu'(e)$ is
  $\bottom$, or $\nu(e)+t$, or $0$ for each old, or inherited, or new event
  $e$, respectively; and $\nu'(\last) = t$.  The most complicated part is
  the definition of $\win(E,t)$ which intuitively corresponds to the
  probability that $E$ is the set of events ``winning'' the competition
  among the events scheduled in~$s$ at time $t$. First, we define a
  ``shifted'' density function $f_{e \mid \nu(e)}$ that takes into account that
  the time $\nu(e)$ has already elapsed. Formally, for a variable-delay event $e$
  and any elapsed time $\nu(e) < u_e$, we define
  \[
  f_{e \mid \nu(e)}(x) = \frac{f_e(x+\nu(e))}{\int_{\nu(e)}^\infty f_e(y) \; \de{y}} \qquad \text{if
    $x \geq 0$.}
  \]
  Otherwise, we define $f_{e\mid \nu(e)}(x) = 0$. The denominator scales the
  function so that $f_{e\mid \nu(e)}$ is again a density function. Finally,  

  \begin{equation*}
    \win(E,t) = 
    \begin{cases}
      f_{e \mid \nu(e)}(t) \cdot \prod_{c\in V \setminus E} \int_t^\infty
      f_{c \mid \nu(c)}(y) \; \de{y}
      & \text{if $E =\{e\} \subseteq V$} \\
      \prod_{c\in V} \int_t^\infty f_{c \mid \nu(c)}(y) \; \de{y} &
      \text{if $E = F' \subseteq F$}\\
      0 &
      \text{otherwise.}
    \end{cases}
  \end{equation*}
  \end{itemize}
A \emph{run} of the Markov chain is an infinite sequence $\sigma = z_0 \;
z_1 \; z_2 \cdots$ of configurations.  The Markov chain is defined on the
probability space $(\allruns, \sigmafield, \probm)$ where $\allruns$ is the
set of all runs, $\sigmafield$ is the product $\sigma$-field
$\Productfield_{i=0}^{\infty} \configsfield$, and $\probm$ is the unique
probability measure such that for every finite sequence $A_0,\cdots,A_n \in
\configsfield$ we have that
\begin{equation*}\label{eq:gssmc}
  \probm(\gssmc_0 {\in} A_0, \cdots,\gssmc_n {\in} A_n) = 
  \int\limits_{z_0 \in A_0}\!\cdots \int\limits_{z_{n-1} \in A_{n-1}}
  \initmeasure(\de{z_0})\cdot \kernel(z_0,\de{z_1}) \cdots
  \kernel(z_{n-1},A_n)
\end{equation*}
where each $\gssmc_i$ is the $i$-th projection of an element in $\allruns$
(the $i$-th configuration of a run).

Finally, we define an $m$-step transition kernel $\kernel^m$ inductively as
$\kernel^1(z,A) = \kernel(z,A)$ and $ \kernel^{i+1}(z,A) = \int_{\Gamma}
\kernel(z,\de{y}) \cdot \kernel^i(y,A)$.

\myspacea
\subsection{Frequency measures}
Our attention focuses on frequencies of a fixed state $\measured \in S$ in the runs of the Markov chain. 
Let $\sigma = (s_0,\nu_0) \; (s_1,\nu_1) \cdots$ be a run. We define
\[
\qquad \qquad 
 \dme(\sigma) = \lim_{n\rightarrow \infty} \frac{\sum_{i=0}^n \delta(s_i)}{n} 
\qquad \qquad
 \cme(\sigma) = \lim_{n\rightarrow \infty} \frac{\sum_{i=0}^n \delta(s_i) \cdot \nu_{i+1}(\last)}{\sum_{i=0}^n \nu_{i+1}(\last)}  
\]
where $\delta(s_i)$ is equal to $1$ when $s_i = \measured$, and $0$ otherwise. 
We recall that $\nu_{i+1}(\last)$ is the time spent in state $s_i$ before moving to $s_{i+1}$.
We say that the random variable $\dme$ or $\cme$ is \emph{well-defined} for a run $\sigma$ if
the corresponding limit exists for $\sigma$. Then, $\dme$ corresponds to the frequency of discrete visits to the state $\measured$ and $\cme$ corresponds to the ratio of time spent in the state $\measured$.

\myspacea
\subsection{Region graph}
In order to state the results in a simpler way, we introduce the \emph{region graph},
a standard notion from the area of timed automata~\cite{AD:Timed-Automata}. It is a 
finite partition of the uncountable set of configurations. 
First, we define the region relation $\region$. For $a,b \in \Rset$, we say that $a$ and $b$ \emph{agree on integral part}
if $\intpart(a) = \intpart(b)$ and neither 
or both $a$, $b$ are integers. Further, we set the bound $B=\max\big(\{\ell_e,u_e\mid e\in\events\}\setminus\{\infty\}\big)$.
Finally, we put $(s_1,\nu_1) \region (s_2,\nu_2)$ if 
\begin{itemize}
\item $s_1 = s_2$;
\item for all $e\in \sched(s_1)$ we have that $\nu_1(e)$ and $\nu_2(e)$
  agree on integral parts or are both greater than $B$;
\item for all $e, f\in\sched(s_1)$ with $\nu_1(e)\leq B$ and $\nu_1(f)\leq B$ we have that $\fr(\nu_1(e)) \leq
  \fr(\nu_1(f))$ iff $\fr(\nu_2(e)) \leq \fr(\nu_2(f))$.
\end{itemize}
Note that $\region$ is an equivalence with finite index. The equivalence
classes of $\region$ are called \emph{regions}. 
We define a finite \emph{region graph} 
$\regiongraph = (V,E)$ where the set of vertices $V$ is the set of regions and for every pair of regions $R,R'$ there is an edge $(R,R') \in E$ iff \mbox{$\kernel(z,R') > 0$} for some $z \in R$.
The construction is correct because all states in the same region have the
same one-step qualitative behavior (for details, see\appref{app:region-graph}).



\myspacea

\section{Two fixed-delay events}
\label{sec-results}
\myspacea 
Now, we explain in more detail what problems can be caused by
fixed-delay events. We start with an example of a GSMP with two fixed-delay
events for which it is not true that the variables $\dme$ and $\cme$ are
well-defined for almost all runs. Then we show some other examples of GSMP
with fixed-delay events that disprove some results from literature.  In the
next section, we provide positive results when the number and type of
fixed-delay events are limited.

\myspaceab

\subsection*{When  the frequencies $\dme$ and $\cme$ are not well-defined}
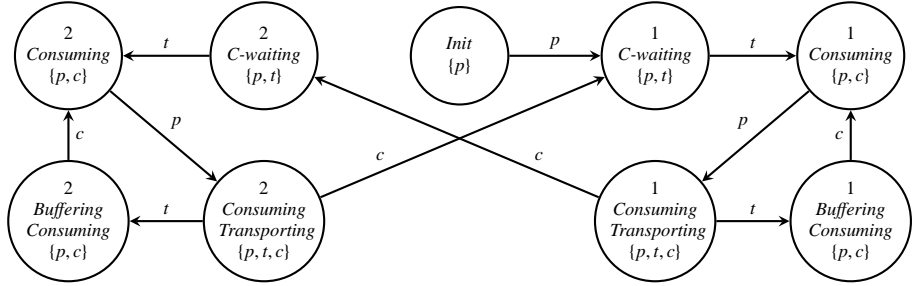
\begin{figure}[t]\myspaceb
  \centering
   
\tikzstyle{max}=[thick,draw,minimum size=1.4em,inner sep=0em]
\tikzstyle{min}=[diamond,thick,draw,minimum size=2em,%
    inner sep=0em]
\tikzstyle{ran}=[circle,thick,draw,text width=12mm,text centered,minimum size=2em,%
    inner sep=0em]
\tikzstyle{act}=[circle,thick,draw,fill,minimum size=.7em,%
    inner sep=0em]
\tikzstyle{mc}=[rounded corners,thick,draw,minimum size=1.4em,%
    inner sep=.5ex]
\tikzstyle{tran}=[thick,draw,->,>=stealth]
\tikzstyle{loop left}=[tran, to path={.. controls +(150:.7)
    and +(210:.7) .. (\tikztotarget) \tikztonodes}]
\tikzstyle{loop right}=[tran, to path={.. controls +(30:.7)
    and +(330:.7) .. (\tikztotarget) \tikztonodes}]
\tikzstyle{loop upright}=[tran, to path={.. controls +(70:.9)
    and +(30:.9) .. (\tikztotarget) \tikztonodes}]
\tikzstyle{loop upleft}=[tran, to path={.. controls +(110:.9)
    and +(150:.9) .. (\tikztotarget) \tikztonodes}]
\tikzstyle{loop above}=[tran, to path={.. controls +(60:.7)
    and +(120:.7) .. (\tikztotarget) \tikztonodes}]
\tikzstyle{loop below}=[tran, to path={.. controls +(240:.7)
    and +(300:.7) .. (\tikztotarget) \tikztonodes}]

 \begin{tikzpicture}[x=1.3cm,y=1.1cm],font=\scriptsize]
\node (init) at (0,0) [ran] { \textit{Init}\\ $\{p\}$};
\node (cwait) at (-2,0) [ran] {2\\\textit{C-waiting}\\ $\{p,t\}$};
\node (C) at (-4,0) [ran] {2\\\textit{Consuming}\\ $\{p,c\}$};
\node (CT) at (-2,-2) [ran] {2\\\textit{Consuming}\\ \textit{Transporting}\\ $\{p,t,c\}$};
\node (BC) at (-4,-2) [ran] {2\\\textit{Buffering}\\ \textit{Consuming} \\$\{p,c\}$};

\node (cwait') at (2,0) [ran] {1\\\textit{C-waiting}\\ $\{p,t\}$};
\node (C') at (4,0) [ran] {1\\\textit{Consuming}\\ $\{p,c\}$};
\node (CT') at (2,-2) [ran] {1\\\textit{Consuming}\\ \textit{Transporting}\\ $\{p,t,c\}$};
\node (BC') at (4,-2) [ran] {1\\\textit{Buffering}\\ \textit{Consuming} \\$\{p,c\}$};

 \draw [tran,rounded corners] (init) 
  edge 
  node[above] {$p$} 
  (cwait');

 \draw [tran,rounded corners] (cwait) 
  edge 
  node[above] {$t$} 
  (C);

 \draw [tran,rounded corners] (C) 
  edge 
  node[above right] {$p$} 
  (CT);

 \draw [tran,rounded corners] (CT) 
  edge 
  node[above] {$t$} 
  (BC);

 \draw [tran,rounded corners] (BC) 
  edge 
  node[right] {$c$} 
  (C);

 \draw [tran,rounded corners] (CT) 
  edge 
  node[above left, near start] {$c$} 
  (cwait');


 \draw [tran,rounded corners] (cwait') 
  edge 
  node[above] {$t$} 
  (C');

 \draw [tran,rounded corners] (C') 
  edge 
  node[above left] {$p$} 
  (CT');

 \draw [tran,rounded corners] (CT') 
  edge 
  node[above] {$t$} 
  (BC');

 \draw [tran,rounded corners] (BC') 
  edge 
  node[left] {$c$} 
  (C');

 \draw [tran,rounded corners] (CT') 
  edge 
  node[above right, near start] {$c$} 
  (cwait);

\end{tikzpicture}
  \caption{A GSMP of a producer-consumer system. The events $p$, $t$, and
$c$ model that a packet production, transport, and consumption is finished,
respectively. 
Below each state label, there is the set of scheduled events. 
The fixed-delay events 
 $p$ and $c$ have $l_p=u_p=l_c=u_c=1$ and the uniformly distributed variable-delay event
    $t$ has $l_t=0$ and $u_t=1$. }
  \label{fig-producer_consumer}\myspaceab
\end{figure}

In Figure~\ref{fig-producer_consumer}, we show an example of a GSMP with two
fixed-delay events and one variable-delay event for which it is not true
that the variables $\dme$ and $\cme$ are well-defined for almost all runs.
It models the following producer-consumer system. We use three components -- a
producer, a transporter and a consumer
of packets. 
The components work in parallel but each component
can process (i.e.~produce, transport, or consume) at most one packet at a
time.

Consider the following time requirements: each packet production takes \emph{exactly} 1~time unit, 
each transport takes \emph{at most} 1~time unit, and each consumption takes again \emph{exactly} 1~time
unit. As there are no limitations to block the producer, it is working for
all the time and new packets are produced precisely each time unit.  As the
transport takes shorter time than the production, every new packet is
immediately taken by the transporter and no buffer is needed at this place.
When a packet arrives to the consumer, the consumption is started
immediately if the consumer is waiting%
; otherwise, the packet is stored into a buffer. When the consumption
is finished and the buffer is empty, the consumer waits; otherwise, a new consumption
starts immediately.

In the GSMP in Figure~\ref{fig-producer_consumer},
the consumer has two modules -- one is in operation and the
other idles at a time -- when the consumer enters the waiting state, it
switches the modules. The labels $1$ and $2$ denote which module of the
consumer is in operation.


One can easily observe that the consumer enters the waiting state (and
switches the modules) if and only if the current transport takes more time
than it has ever taken. As the transport time is bounded by $1$, it gets
harder and harder to break the record. As a result, the system stays in the
current module on average for longer time than in the previous
module. Therefore, due to the successively prolonging stays in the modules,
the frequencies for 1-states and 2-states oscillate. For precise
computations, see\appref{app:ex:not-well-def}. We conclude the above
observation by the following theorem.

\begin{theorem}\label{thm-two-events-no-limit}
  There is a GSMP (with two fixed-delay events and one variable-delay event)
  for which it is \emph{not} true that the variables $\cme$ and $\dme$ are
almost surely
  well-defined.
\end{theorem}

\myspaceb

\subsection*{Counterexamples}

In~\cite{ACD:model-checking-real-time,ACD:verifying-automata-real-time}
there are algorithms for GSMP model checking based on the region
construction. They rely on two crucial statements of the papers: 
\begin{enumerate}
\item Almost all runs end in some of the bottom strongly connected components
  (BSCC) of the region graph.
\item Almost all runs entering a BSCC 
 visit all regions of the component infinitely often.
\end{enumerate}

Both of these statements are true for finite state Markov chains. In the
following, we show that neither of them has to be valid for region graphs of GSMP.

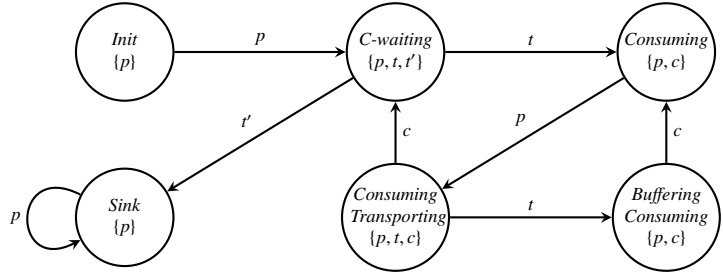
\begin{figure}[t]
  \centering
   
\tikzstyle{max}=[thick,draw,minimum size=1.4em,inner sep=0em]
\tikzstyle{min}=[diamond,thick,draw,minimum size=2em,%
    inner sep=0em]
\tikzstyle{ran}=[circle,thick,draw,text width=12mm,text centered,minimum size=2em,%
    inner sep=0em]
\tikzstyle{act}=[circle,thick,draw,fill,minimum size=.7em,%
    inner sep=0em]
\tikzstyle{mc}=[rounded corners,thick,draw,minimum size=1.4em,%
    inner sep=.5ex]
\tikzstyle{tran}=[thick,draw,->,>=stealth]
\tikzstyle{loop left}=[tran, to path={.. controls +(140:1.1)
    and +(220:1.1) .. (\tikztotarget) \tikztonodes}]
\tikzstyle{loop right}=[tran, to path={.. controls +(30:.7)
    and +(330:.7) .. (\tikztotarget) \tikztonodes}]
\tikzstyle{loop upright}=[tran, to path={.. controls +(70:.9)
    and +(30:.9) .. (\tikztotarget) \tikztonodes}]
\tikzstyle{loop upleft}=[tran, to path={.. controls +(110:.9)
    and +(150:.9) .. (\tikztotarget) \tikztonodes}]
\tikzstyle{loop above}=[tran, to path={.. controls +(60:.9)
    and +(120:.9) .. (\tikztotarget) \tikztonodes}]
\tikzstyle{loop below}=[tran, to path={.. controls +(240:.7)
    and +(300:.7) .. (\tikztotarget) \tikztonodes}]

 \begin{tikzpicture}[x=1.8cm,y=1.1cm],font=\scriptsize]
\node (init) at (0,0) [ran] { \textit{Init}\\ $\{p\}$};
\node (cwait) at (2,0) [ran] {\textit{C-waiting}\\ $\{p,t,t'\}$};
\node (C) at (4,0) [ran] {\textit{Consuming}\\ $\{p,c\}$};
\node (CT) at (2,-2) [ran] {\textit{Consuming}\\ \textit{Transporting}\\ $\{p,t,c\}$};
\node (BC) at (4,-2) [ran] {\textit{Buffering}\\ \textit{Consuming} \\$\{p,c\}$};
\node (restart) at (0,-2) [ran] { \textit{\restart}\\ $\{p\}$};

 \draw [tran,rounded corners] (init) 
  edge 
  node[above] {$p$} 
  (cwait);

 \draw [tran,rounded corners] (cwait) 
  edge 
  node[above left] {$t'$} 
  (restart);

 \draw [tran,rounded corners] (restart) 
  edge [loop left] 
  node[left] {$p$} 
  (restart);

 \draw [tran,rounded corners] (cwait) 
  edge 
  node[above] {$t$} 
  (C);

 \draw [tran,rounded corners] (C) 
  edge 
  node[above left] {$p$} 
  (CT);

 \draw [tran,rounded corners] (CT) 
  edge 
  node[above] {$t$} 
  (BC);

 \draw [tran,rounded corners] (BC) 
  edge 
  node[right] {$c$} 
  (C);

 \draw [tran,rounded corners] (CT) 
  edge 
  node[right] {$c$} 
  (cwait);

\end{tikzpicture}
  \caption{A GSMP with two fixed-delay events $p$ and $c$ (with
      $l_p=u_p=l_c=u_c=1$), a uniformly distributed variable-delay events $t$, $t'$
      (with $l_t=l_{t'}=0$ and $u_t=u_{t'}=1$).}
  \label{fig-producer_consumer_alur}\myspaceb
\end{figure}
Let us consider the GSMP depicted in
Figure~\ref{fig-producer_consumer_alur}.
This is a
producer-consumer model similar to the previous example but we have only one
module of the consumer here. 
Again, 
entering the state \emph{C-waiting} indicates that the current
transport takes more time than it has ever taken. 
In the state \textit{C-waiting}, an additional event $t'$ can occur and move the system into a state
\textit{\restart}.
%
One can intuitively observe that 
we enter the state \emph{C-waiting} less and less often
and stay there for shorter and shorter time. 
Hence, the probability that the event
$t'$ occurs in the state \emph{C-waiting} is decreasing 
during the run. 
For precise computations proving the following claim, see\appref{app:ex:bscc}.

\begin{claim}
The probability to reach \emph{\restart} from \emph{Init} is strictly less than $1$.
\end{claim}
The above claim directly implies the following theorem thus disproving statement 1.
\begin{theorem}\label{thm:nespadnu}
  There is a GSMP (with two fixed-delay and two variable delay
  events) where the probability to reach any BSCC of the region graph is strictly smaller than 1.
\end{theorem}

Now consider in Figure~\ref{fig-producer_consumer_alur} a transition under the event $p$ from the state \emph{\restart} 
to the state \emph{Init} instead of the self-loop. This turns the whole region graph into a single BSCC.
We prove that the state \emph{\restart} is almost surely visited only finitely often.
Indeed, let $p<1$ be the original probability to reach \emph{\restart} guaranteed by the claim above.
The probability to reach \emph{\restart} from \emph{\restart} again is also
$p$ as the only transition leading from  \emph{\restart} enters the initial configuration. Therefore, the probability to reach \emph{\restart} infinitely often is $\lim_{n\to\infty}p^n=0$. This proves the following theorem. Hence, the statement 2 of~\cite{ACD:model-checking-real-time,ACD:verifying-automata-real-time} is disproved, as well.

\begin{theorem}\label{thm:nevymetam}
  There is a GSMP (with two fixed-delay and two variable delay events) 
with strongly connected region graph and with a region that is reached infinitely often
with probability $0$.
\end{theorem}


\myspacea

\section{Single-ticking GSMP}
\label{sec:one-results}


%


First of all, motivated by the previous counterexamples, we identify the
behavior of the fixed-delay events that may cause $\dme$ and $\cme$ to be
undefined. The problem lies in fixed-delay events that can immediately
schedule themselves whenever they occur; such an event can occur
periodically like ticking of clocks.  In the example of
Figure~\ref{fig-producer_consumer_alur}, there are two such events $p$ and
$c$. The phase difference of their ticking gets smaller and smaller, causing
the unstable behavior. 




For two fixed-delay events $e$ and $e'$, we say that $e$ \emph{causes} $e'$
if there are states $s$, $s'$ and a set of events $E$ such that
$\occur(s,E)(s')\!>\!0$, $e\in E$, and $e'$ is newly scheduled in $s'$.

\begin{definition}
A GSMP is called \emnice if either there is no fixed-delay event or
there is a strict total order $<$ on fixed-delay events with the least element $e$ (called \emph{ticking} event) such that
whenever $f$ causes $g$ then either $f<g$ or $f=g=e$.
\end{definition}
From now on we restrict 
to \nice GSMP and 
prove our main positive result.

\begin{theorem}\label{thm:main1}
  In \nice GSMP, the random variables $\dme$ and $\cme$ are well-defined for almost every
  run and admit only finitely many values. Precisely, almost every run reaches a BSCC of the region graph and for each BSCC $B$ there are values $d, c \in [0,1]$ such that
  $\dme(\sigma) = d$ and $\cme(\sigma) = c$ for almost all runs $\sigma$ 
  that reach the BSCC $B$.
\end{theorem}

The rest of this section is devoted to the proof of Theorem~\ref{thm:main1}. 
First, we show that almost all runs 
end up trapped in some BSCC of the region graph. Second, we solve the problem 
while restricting to runs that \emph{start} in a BSCC (as the initial 
part of a run outside of any BSCC is not relevant for the long run average behavior). We show that in a BSCC, the variables $\dme$ and $\cme$
are almost surely constant. 
The second part of the proof relies on several standard results from 
the theory of general state space Markov chains.
Formally, the proof follows from Propositions~\ref{prop:reach-bscc} and~\ref{prop:bscc-not-sync} stated below.

\myspacea

\subsection{Reaching a BSCC}

\newcommand{\PropReachBscc}{
In \nice GSMP, almost every run reaches a BSCC of the region graph.}

\begin{proposition}
\label{prop:reach-bscc}
\PropReachBscc
\end{proposition}

The proof uses similar methods as the proof in~\cite{ACD:verifying-automata-real-time}.
%
By definition, the process moves along the edges of the region graph.
From every region, there is a minimal path through the region graph into a BSCC, let $n$ be the maximal length of all such paths. 
Hence, in at most $n$ steps the process reaches a BSCC with positive probability from any configuration.
Observe that if this probability was bounded from below, we would eventually reach a BSCC from any configuration almost surely.
%
However, this probability can be arbitrarily small.
Consider the following example with event $e$ uniform
on $[0,1]$ and event $f$ uniform on $[2,3]$. In an intuitive notation, let $R$ be the region $[0 < e < f < 1]$. What is the probability that the event $e$ occurs after the elapsed time of $f$ reaches $1$ (i.e. that the region $[e = 0; 1 < f < 2]$ is reached)? For a configuration in $R$ with valuation $((e,0.2),(f,0.7))$ the probability is $0.5$ but for another configuration in $R$ with $((e,0.2),(f,0.21))$ it is only $0.01$. 
Notice that the transition probabilities depend on the difference of the fractional values
of the clocks, we call this difference \emph{separation}. Observe that in other situations, the separation of clocks from value $0$ also matters.

\begin{definition}\label{def:delta-separation}
Let $\delta > 0$. We say that a configuration $(s,\nu)$ is \emph{$\delta$-separated} if for every $x,y\in\{ 0 \} \cup \{ \nu(e) \mid e \in \sched(s)\}$, 
we have either $|\fr(x) - \fr(y)| > \delta$ or $\fr(x) = \fr(y)$.
\end{definition}

We fix a $\delta > 0$. To finish the proof using the concept of $\delta$-separation, we need two observations. 
First, from \emph{any} configuration we reach in $m$ steps a $\delta$-separated configuration with probability at least $q > 0$. Second, the probability to reach a fixed region from \emph{any} $\delta$-separated configuration is bounded from below by some $p > 0$.
%
By repeating the two observations ad infinitum, we reach some BSCC almost surely. Let us state the claims. For proofs, see\appref{app:reach}.

\newcommand{\LemDeltaSep}{
There is $\delta>0$, $m \in \Nset$ and $q > 0$ such that from every configuration we reach a 
$\delta$-separated configuration in $m$ steps with probability at least $q$.
}

\begin{lemma}
\label{lem:delta-sep}
\LemDeltaSep
\end{lemma}

\newcommand{\LemReach}{
For every $\delta > 0$ and $k \in \Nset$ there is $p > 0$ such that for any pair of regions $R$, $R'$ connected by a path of length $k$ and
for any $\delta$-separated $z \in R$, we have $\kernel^k(z,R') > p$.
}

\begin{lemma}
\label{lem:reach}
\LemReach
\end{lemma}
%
%

Lemma~\ref{lem:reach} holds even for unrestricted GSMP. Notice that Lemma~\ref{lem:delta-sep} does not.
As in the example of Figure~\ref{fig-producer_consumer_alur}, the separation may be non-increasing for all runs.

\myspacea

\subsection{Frequency in a BSCC}

From now on, we deal with the bottom strongly connected components that are reached almost surely.
Hence, we assume that the region graph $\regiongraph$ is strongly connected.
We have to allow an arbitrary initial configuration $z_0 = (s,\nu)$; in particular, $\nu$ does not have to be a zero vector.\footnote{Technically, the initial measure is $\initmeasure(A) = 1$ if $z_0 \in A$ and $\initmeasure(A) = 0$, otherwise.}

\newcommand{\PropBsccNotSync}{
In a \nice GSMP with strongly connected region graph, there are values $d,c \in [0,1]$ 
such that for any initial configuration $z_0$ 
and for almost all runs $\sigma$ starting from $z_0$, we have that
$\dme$ and $\cme$ are well-defined and $\dme(\sigma) = d$ and $\cme(\sigma) = c$. 
}

\begin{proposition}\label{prop:bscc-not-sync}
\PropBsccNotSync
\end{proposition}

We assume that the region graph is aperiodic in the following sense. 
A \emph{period} $p$ of a graph $G$ is the greatest common divisor of lengths of all cycles in $G$. 
The graph $G$ is \emph{aperiodic} if $p=1$.
Under this assumption\footnote{%
If the region graph has period $p > 1$, we can employ the standard technique 
and decompose the region graph (and the Markov chain) into $p$ aperiodic components. 
The results for individual components yield straightforwardly the results for the whole Markov chain, see, e.g., \cite{BKKKR-HSSC2011}.}, the chain $\gssmc$ is in some sense stable. Namely, 
(i) $\gssmc$ has a unique invariant measure that is independent of the initial measure and (ii) the strong law of large numbers (SLLN) holds for $\gssmc$. 

First, we show that (i) and (ii) imply the proposition. Let us recall the notions.
%
%
%
%
We  say that a probability measure $\pi$ on $(\configs,\configsfield)$ is \emph{invariant} if for all $A\in\configsfield$ 
    \[
    \pi(A)\quad = \quad \int_{\configs} \pi(\de{x})\kernel(x,A).
    \]

\noindent
The SLLN states that if $h:\configs \rightarrow \Rset$ satisfies 
$E_\pi[h]<\infty$, 
then almost surely
    \begin{align}\label{eq:ssln}
    \lim_{n\rightarrow \infty} \frac{\sum_{i=1}^n h(\Phi_i)}{n}\quad = \quad
    E_\pi[h],
    \end{align}
where $E_\pi[h]$ is the expected value of $h$ according to the invariant measure $\pi$.

We set $h$ as follows. For a run $(s_0,\nu_0)(s_1,\nu_1)\cdots$, let
$h(\Phi_i) = 1$ if $s_i=\measured$ and $0$, otherwise.  We have
$E_\pi[h] < \infty$ since $h \leq 1$.  From (\ref{eq:ssln}) we
obtain that almost surely
\begin{align*}\label{eq:freq-def}
    \dme \quad = \quad \lim_{n\rightarrow \infty} \frac{\sum_{i=1}^n h(\Phi_i)}{n}\quad 
= \quad
    E_\pi[h].
\end{align*}
As a result, $\dme$ is well-defined and equals the constant value $E_\pi[h]$ for almost all runs. 
We treat the variable $\cme$ similarly. 
Let $W((s,\nu))$ denote the expected waiting time of the GSMP in the configuration $(s,\nu)$.
We use a function $\tau((s,\nu)) = W((s,\nu))$ if $s = \measured$ and $0$, otherwise.
Since all the events have finite expectation, 
we have $E_\pi[W] < \infty$ and $E_\pi[\tau] < \infty$.
Furthermore, we show in\appref{app:frequencies} that almost surely
\begin{align*}
  \cme \quad = \quad  \lim_{n\rightarrow \infty} \frac{\sum_{i=1}^n \tau(\Phi_i)}{\sum_{i=1}^n W(\Phi_i)}
\quad = \quad 
    \frac{E_\pi[\tau]}{E_\pi[W]}.
\end{align*}
Therefore, $\cme$ is well-defined and 
equals the constant $E_\pi[\tau]/E_\pi[W]$ for almost all runs.

Second, we prove (i) and (ii). 
A standard technique of general state space Markov chains (see, e.g.,
\cite{MT:book}) yields (i) and (ii) for chains that satisfy the following
condition. Roughly speaking, we search for a set of configurations $C$ that
is visited infinitely often and for some $\ell$ the measures
$\kernel^\ell(x,\cdot)$ and $\kernel^\ell(y,\cdot)$ are very similar for any
$x,y\in C$.  This is formalized by the following lemma.


\newcommand{\LemSmall}{
There is a measurable set of configurations $C$ such that
\begin{enumerate}
 \item there is $k \in \Nset$ and $\alpha > 0$ such that for every $z \in \configs$ we have $\kernel^{k}(z,C) \geq \alpha$, and
 \item there is $\ell \in \Nset$, $\beta > 0$, and a probability measure $\smallmeasure$ such that for every $z \in C$ and $A \in \configsfield$
we have $\kernel^{\ell}(z,A) \geq \beta \cdot \smallmeasure(A)$.
\end{enumerate}
}

\begin{lemma}\label{lem:small}
\LemSmall
\end{lemma}


\begin{proof}[Sketch]
%
%
Let $e$ be the ticking event and $R$ some reachable region where $e$ is the event closest to its upper bound. We fix a sufficiently small $\delta > 0$ and choose $C$ to be the set of $\delta$-separated configurations of $R$. We prove the first part of the lemma similarly to Lemmata~\ref{lem:delta-sep} and~\ref{lem:reach}.
As regards the second part, we define the measure $\smallmeasure$ uniformly
on a hypercube $X$ of configurations $(s,\nu)$ that have $\nu(e) = 0$ and
$\nu(f) \in (0,\delta)$, for $f\neq e$. First, assume that $e$ is the only
fixed-delay event. We fix $z=(s',\nu')$ in $R$; let $d = u_e-\nu'(e)>
\delta$ be the time left in $z$ before $e$ occurs. For simplicity, we assume
that each variable-delay events can occur after an arbitrary delay $x \in
(d-\delta, d)$. Precisely, that it can occur in an
$\varepsilon$-neighborhood of $x$ with probability bounded from below by
$\beta \cdot \varepsilon$ where $\beta$ is the minimal density value of all
$\events$.  Note that the variable-delay events can be ``placed'' this way
arbitrarily in $(0,\delta)$. Therefore, when $e$ occurs, it has value $0$
and all variable-delay events can be in interval $(0,\delta)$.  In other
words, we have $\kernel^{\ell}(z,A) \geq \beta \cdot \smallmeasure(A)$ for
any measurable $A \subseteq X$ and for $\ell = |\events|$.

Allowing other fixed-delay events causes some trouble because a fixed-delay event $f \neq e$ cannot be ``placed'' arbitrarily. In the total order $<$, the event $f$ can cause only strictly greater fixed-delay events. The greatest fixed-delay event can cause only variable-delay events that can be finally ``placed'' arbitrarily as described above.\QED
\end{proof}

\myspaceab

\section{Approximations}


In the previous section we have proved that in \nice GSMP, $\dme$ and $\cme$ are almost surely well-defined and for almost all runs they attain only finitely many values $d_1\ldots,d_k$ and $c_1,\ldots,c_k$, respectively. In this section we show how to approximate $d_i$'s and $c_i$'s and the probabilities that $\dme$ and $\cme$ attain these values, respectively. 

\newcommand{\claimApproxReach}{
Let $X$ be a set of all configurations in a BSCC $\mathcal B$, $X_\measured\subseteq X$ the set of configurations with state $\measured$, and $d$ the frequency corresponding to $\mathcal B$. There are computable constants 
$n_1,n_2 \in \Nset$ and $p_1,p_2 > 0$ such that for every $i \in \Nset$ and $z_X\in X$ we have
\[\begin{array}{rll}
   \lvert \probm(\reach(X)) - \kernel^i(z_0,X) \rvert \; &\leq \; (1-p_1)^{\lfloor i/n_1 \rfloor} & \\
   \lvert d - \kernel^i(z_X,X_\measured)\rvert \; & \leq \; (1-p_2)^{\lfloor i/n_2 \rfloor} 
  \end{array}
\]
}

\newcommand{\claimFunctionW}{
On each region, $W$ is continuous, and $E_\pi[W]$ is finite.
}

\begin{theorem}\label{thm:approx}
In a \nice GSMP, let $d_1,\ldots,d_k$ and $c_1,\ldots,c_k$ be the discrete and timed frequencies, respectively, corresponding to BSCCs of the region graph. For all $1\leq i\leq k$, the numbers $d_i$ and $c_i$ as well as the probabilities $\probm(\dme=d_i)$ and $\probm(\cme=c_i)$
can be 
approximated up to any $\varepsilon > 0$.
\end{theorem}

\begin{proof}

Let $X_1, \ldots, X_k$ denote the sets of configurations in individual BSCCs and $d_i$ and $c_i$ correspond to $X_i$. 
Since we reach a BSCC almost surely, we have
\[
 \probm(\dme = d_i) = \sum_{j=1}^k \probm(\dme = d_i \mid \reach(X_j)) \cdot \probm(\reach(X_j))=\sum_{j=1}^k \indicator{d_j=d_i} \cdot \probm(\reach(X_j))
\] where the second equality follows from the fact that
almost all runs in the $j$-th BSCC yield the discrete frequency $d_j$.
Therefore, $\probm(\dme = d_i)$ and $d_i$ can be approximated as follows using the methods of~\cite{RR:GSSMC-PS}.

 
\begin{claim}\label{claim:approx-reach}
\claimApproxReach
\end{claim}
Further, we want to approximate $c_i=E_\pi[\tau]/E_\pi[W]$, where $\pi$ is the invariant measure on $X_i$. In other words, we need to approximate $\int_{X_i} 
\tau(x)\pi(dx)$ and $\int_{X_i} 
W(x)\pi(dx)$. 
An $n$-th approximation $w_n$ of $E_\pi[W]$ can be gained by discretizing the part of the state space $\{(s,\nu)\in \Gamma \mid \forall e \in \sched(s): \nu(e) \leq n\}$ into, e.g., $1/n$-large hypercubes, where the invariant measure $\pi$ is approximated using $\kernel^n$.
This approximation converges to $E_\pi[W]$ since $W$ is continuous and $E_\pi[W]$ is finite.
For the details of the following claim, see\appref{sec:app-results}.

%
%
%

\begin{claim}
 \claimFunctionW
\end{claim}

This concludes the proof as 
$\tau$ only differs from $W$ in being identically zero on some
regions; thus, $E_\pi[\tau]$ can be approximated analogously.

\myspacea

%
%
%
%

\end{proof}


\myspacea\myspacea

\section{Conclusions, future work}


We have studied long run average properties of generalized semi-Markov
processes with both fixed-delay and variable-delay events.  We have shown
that two or more (unrestricted) fixed-delay events lead to considerable
complications regarding stability of GSMP. In particular, we have shown that
the frequency of states of a GSMP may not be well-defined and that bottom
strongly connected components of the region graph may not be reachable with
probability one. This leads to counterexamples disproving several results
from literature. On the other hand, for \nice GSMP we have proved
that the frequencies of states are well-defined for almost all
runs. Moreover, we have shown that almost every run has one of finitely many
possible frequencies that can be effectively approximated (together with
their probabilities) up to a given error tolerance.

In addition, the frequency measures can be easily extended into the mean payoff setting. Consider assigning real rewards to states. The mean payoff then corresponds to the frequency weighted by the rewards.

Concerning future work, the main issue is efficiency of algorithms for
computing performance measures for GSMP. We plan to work on both better
analytical methods as well as 
practicable approaches to Monte Carlo
simulation. One may also consider extensions of our positive results to
controlled GSMP and games on GSMP.

\myspacea

\bibliographystyle{splncs03}
\bibliography{str-long,concur}

\appendix

\newpage
\section{Details on counterexamples}
\label{sec:app-counter}

\newcommand{\XS}{\mathit{S}}
\newcommand{\XD}{D}
\newcommand{\Xeps}{\alpha}
\newcommand{\sestnact}{{3}}
\newcommand{\sedmnact}{{3}}
\newcommand{\hl}[1]{L(#1)}

\begin{definition}
A \emph{distance} of two events $e$ and $f$ (in this order) in a configuration $(s,\nu)$ is $\fr(\nu(f)-\nu(e))$ .
\end{definition}

\subsection{When the frequencies $\dme$ and $\cme$ are not well-defined: Proof of Theorem~\ref{thm-two-events-no-limit}}\label{app:ex:not-well-def}
In the following, we prove that in our example $\dme$ and $\cme$ are not well-defined for almost all runs. Namely, that there is a set of runs with positive measure such that for these runs the partial sums oscillate.

After setting the initial distance of events $p$ and $c$, every run stays in the 1-states (labeled with 1) until the distance is lessened in the state \emph{2 C-waiting}. This sojourn in the 1-states is called the first \emph{phase}. Then the run continues with the second phase now in the 2-states until the distance is lessened again and it moves back to 1-states and begins the third phase etc. Each phase consists of repeating several \emph{attempts}, i.e.~running through the cycle of length three. In each attempt the distance gets smaller with probability $d$ (where $d$ is the current distance) and stays the same with probability $1-d$ due to the uniform distribution of $t$. This behavior corresponds to the geometric distribution. The density on the new distance is uniform on the whole $d$. A phase is called \emph{strong} if the newly generated distance is at most half of the old one. Further, we define a \emph{half-life} to be a maximum continuous sequence of phases where exactly the last one is strong. Every run can thus be uniquely decomposed into a sequence of half-lives. The random variable stating the distance at the beginning of the $j$-th phase of the $i$-th half-life is denoted $\XD_{i,j}$. Denoting the number of phases in the $i$-th half-life by $\hl{i}$ we get $\XD_{n-1,\hl{i}}\geq2\XD_{n,1}$. Thus by induction, we have for all $n,i\in\Nset$ and $j\leq\hl{n-i}$,
\begin{equation}\label{eq:decka}
 \XD_{n-i,j}\geq 2^i\cdot\XD_{n,1}
\end{equation} 
Further, let $\XS_{i,j}$ be the number of attempts in the $j$-th phase of the $i$-th half-life, i.e.~a \emph{length} of this phase. We can now prove the following lemma. Roughly speaking, there are runs (of overall positive measure) where some phase is longer than the overall length of all phases up to that point. Note that the precise statement of the lemma implies moreover that this happens even infinitely often on runs of overall positive measure.

\begin{lemma}
There are $\Xeps>0$ and $m>0$, such that for every $n>1$ there is a set $\setofruns_n$ of measure at least $m$ of runs satisfying $$\XS_{n,1}\geq \Xeps \sum_{\substack{i=1..n-1\\j=1..\hl{i}}}\XS_{i,j}$$
\end{lemma}
\begin{proof}
We set $\Xeps=2/(\sedmnact\cdot(6+2\cdot\sestnact))=1/18$ and $m=1/4$ and let $n>1$ be arbitrary. We define the set $\setofruns_n$ to be the set of all runs $\sigma$ such that the following conditions hold:
\begin{enumerate}
 \item $\XS_{n,1}>1/(2\XD_{n,1})$, \\
(the length of the ``last'' phase is above its expecation),
 \item for all $1\leq i<n$, $\hl{i}\leq (n-i)+\sestnact$, \\
(previous half-lives have no more phases than $n+2,n+1, \ldots, 5, 4$, respectively),
 \item for all $1\leq i<n$ and $1\leq j\leq \hl{i}$, $\XS_{i,j}\leq\sedmnact (n-i)/D_{i,j}$,\\
(all phases in previous half-lives are short w.r.t their expectations).
\end{enumerate}

Denote $\XD:=\XD_{n,1}$  We firstly prove that $\XS_{n,1}\geq \Xeps \sum_{i=1..n-1,j=1..\hl{i}}\XS_{i,j}$ for all runs in $\setofruns_n$. Due to the inequality~(\ref{eq:decka}) and requirements 2.\ and 3., we can bound the overall length of all previous phases by 
$$
\sum_{\substack{i=1..n-1\\j=1..\hl{i}}}\XS_{i,j}
\leq
\sum_{i=1}^{n-1}\frac{(i+\sestnact)\cdot \sedmnact i}{2^i\cdot \XD}
\leq
\sum_{i=1}^{\infty}\frac{(i+\sestnact)\cdot \sedmnact i}{2^i\cdot \XD}
=
\frac{\sedmnact(6+2\cdot\sestnact)}{\XD}=\frac{1}{2\Xeps\XD}$$ and conclude by the requirement 1.

It remains to prove that measure of $\setofruns_n$ is at least $m$. We investigate the measures of the runs described by requirements 1.--3. Firstly, the probability that $\XS_{n,1}>\frac{1}{2\XD_{n,1}}$ is $(1-\XD_{n,1})^{1/2\XD_{n,1}}$, which approaches $1/\sqrt{e}$ as $n$ approaches infinity and is thus greater than $1/2$ for $\XD_{n,1}\leq 1/2$, i.e.~for $n\geq2$. Out of this set of runs of measure $1/2$ we need to cut off all runs that do not satisfy requirements 2.~or 3. As for 2., the probability of $i$-th half-life failing to satisfy 2.~is $(1/2)^{(n-i)+\sestnact}$ corresponding to at least $(n-i)+\sestnact$ successive non-strong phases. Therefore, 2.~cuts off  $\sum_{i=1}^{n-1}1/2^{(n-i)+\sestnact} = \sum_{i=1}^{n-1}1/2^{i+\sestnact} \leq \sum_{i=1}^\infty1/2^{i+\sestnact}=1/2^\sestnact$. 
From the remaining runs we need to cut off all runs violating 3. 
Since the probability of each $\XS_{i,j}$ failing is $(1-\XD_{i,j})^{\sedmnact (n-i)/\XD_{i,j}}$, 
the overall probability of all violating runs is due to~\ref{eq:decka} at most 
\begin{align*}
\sum_{i=1}^{n-1} \sum_{j=1}^{L(i)} (1-\XD_{i,j})^{\sedmnact (n-i)/\XD_{i,j}}
= & \;
\sum_{i=1}^{n-1} \sum_{j=1}^{L(n-i)} (1-\XD_{n-i,j})^{\sedmnact i/\XD_{n-i,j}}
\leq 
\sum_{i=1}^{n-1} \sum_{j=1}^{L(n-i)} (1-2^i\XD)^{\sedmnact i/2^i\XD} \\
\leq & \;
\sum_{i=1}^{n-1} (i+\sestnact) (1-2^i\XD)^{\sedmnact i/2^i\XD}
\leq
\sum_{i=1}^\infty(i+\sestnact)(1/e)^{\sedmnact i} \\
= & \;
\frac{4 e^3-3}{(e^3-1)^2}
< 
1/4
\end{align*}
Altogether the measure of $\setofruns_n$ is at least $m=1/2-1/8-1/4=1/8$.\QED
\end{proof}

Due to the previous lemma, moreover, there is a set $\setofruns$ of runs of positive measure such that each run of $\setofruns$ is contained in infinitely many $\setofruns_n$'s. 

Let us measure the frequency of 1-states (we slightly abuse the notation and denote by $\dme(\sigma)$ and $\cme(\sigma)$ the sum of frequencies of all 1-states instead of one single state $\measured$). We prove that neither $\dme(\sigma)$ nor $\cme(\sigma)$ is well-defined on any $\sigma\in\setofruns$. Since attempts last for one time unit, non-existence of $\dme(\sigma)$ implies non-existence of $\cme(\sigma)$. Thus, assume for a contradiction that $\dme(\sigma)$ is well-defined. Denote $s_i$ the number of attempts in the $i$-th phase. Because 1-states are visited exactly in odd phases, we have $$\dme(\sigma)=\lim_{n\to\infty}\frac{\sum_{i=1}^n s_i\cdot odd(i)}{\sum_{i=1}^n s_i}$$ where $odd(i)=1$ if $i$ is odd and 0 otherwise. By the definition of limit, for every $\varepsilon>0$ there is $n_0$ such that for all $n>n_0$ 
\begin{equation}\label{eq:rozdil}
 \left|\frac{\sum_{i=1}^n s_i\cdot odd(i)}{\sum_{i=1}^n s_i}-\frac{\sum_{i=1}^{n-1} s_i\cdot odd(i)}{\sum_{i=1}^{n-1} s_i}\right|<\varepsilon
\end{equation}  Due to the lemma, $s_n\geq \Xeps \sum_{i=1..n-1}s_i$ happens for infinitely many both odd and even phases $n$ on $\sigma\in\setofruns$. 
Now let $\dme(\sigma)\leq 1/2$, the other case is handled symmetrically. Let $\varepsilon$ be such that $\Xeps\geq\frac{\varepsilon}{1-2\varepsilon-\dme(\sigma)}$, and we choose an odd $n>n_0$ satisfying $s_n\geq \Xeps \sum_{i=1}^{n-1}s_i \geq 
\frac{\varepsilon}{1-2\varepsilon-\dme(\sigma)} \sum_{i=1}^{n-1}s_i$. Denoting $A = \sum_{i=1}^{n-1} s_i$ and $O = \sum_{i=1}^{n-1} s_i\cdot odd(i)$ we get from~(\ref{eq:rozdil}) that
\begin{align*}
\frac{O + s_n}{A + s_n} - \frac{O}{A} 
&\;\; \geq \;\;
\frac{O + \frac{\varepsilon}{1-2\varepsilon-\dme(\sigma)} A}{A + \frac{\varepsilon}{1-2\varepsilon-\dme(\sigma)} A} - \frac{O}{A}
\;\; \stackrel{(\ast)}{=} \;\;
\frac{\varepsilon}{1-\dme(\sigma) - \varepsilon} \cdot \left( 1 - \frac{O}{A} \right) \\
&\; \stackrel{(\ast\ast)}{\geq} \;
\frac{\varepsilon}{1-\dme(\sigma)-\varepsilon} \cdot (1-\dme(\sigma)-\varepsilon) 
\;\; = \;\;
\varepsilon
\end{align*}
which is a contradiction with~(\ref{eq:rozdil}).
Notice that we omitted the absolute value from~(\ref{eq:rozdil}) because for an odd $n$ the term is non-negative. The equality $(\ast)$ is a straightforward manipulation. In $(\ast\ast)$ we use, similarly to (\ref{eq:rozdil}), that $|\frac{O}{A} - \dme(\sigma)| < \varepsilon$.
%

\subsection{Counterexamples: Proof of Claim}\label{app:ex:bscc}
\newcommand{\osmnact}{2}

In the following, we prove that the probability to reach the state \emph{\restart} is strictly less than $1$.

Similarly as in the proof of Theorem~\ref{thm-two-events-no-limit}, we introduce phases and half-lives and proceed with similar but somewhat simpler arguments. Let $d$ be the distance of events $p$ and $c$. Note that $1-d$ is the maximum length of transportation so far. The initial distance is generated in the state \emph{C-waiting} with a uniform distribution on $(0,1)$. After that, the distance gets smaller and smaller over the time (if we ignore the states where the distance is not defined) whenever we enter the state \emph{C-waiting}. Each sequence between two successive visits of \emph{C-waiting} on a run is called a \emph{phase} of this run. After each phase the current distance is lessened. The density on the new distance is uniform on the whole $d$. A phase is called \emph{strong} if the newly generated distance is at most half of the old one. Further, we define a \emph{half-life} to be a maximum continuous sequence of phases where exactly the last one is strong. Every run can thus be uniquely decomposed into a sequence of half-lives (with the last segment being possibly infinite if \emph{C-waiting} is never reached again). The random variable stating the distance at the beginning of the $i$-th half-life is denoted by $\XD_{i}$. By definition, $\XD_{i}\leq\XD_{i-1}/2$ and by induction, for every run with at least $i$ half-lives
\begin{equation}\label{eq:decka2}
 \XD_{i}\leq 1/2^i\,.
\end{equation} 
Denoting the number of phases in the $i$-th half-life by $\hl{i}$, we can prove the following lemma. 
\begin{lemma}
There is $m>0$ such that for every $n>1$ the set $\setofruns_n$ of runs $\sigma$ satisfying
\begin{enumerate}
 \item $\sigma$ does not visit \emph{\restart} during the first $n$ half-lives, and
 \item for every $1\leq i\leq n$ not exceeding the number of half-lives of $\sigma$, $\hl{i}(\sigma)\leq \osmnact\cdot i$
\end{enumerate}
has measure at least $m$.
\end{lemma}
This lemma concludes the proof, as there is a set of runs of measure at least $m$ that never reach the state \emph{\restart}. We now prove the lemma.

Firstly, for every $n$ we bound the measure of runs satisfying the second condition. The probability that $\osmnact i$ consecutive phases are not strong, i.e.~$\hl{i}>\osmnact i$, is $1/2^{\osmnact i}$ as $t$ is distributed uniformly. Therefore, the probability that there is $i\leq n$ with $\hl{i}>\osmnact i$ is less than $\sum_{i=1}^n 1/2^{\osmnact i}$. This probability is thus for all $n\in\Nset$ less than $\sum_{i=1}^\infty 1/2^{\osmnact i}=1/3$. Hence, for each $n$ at least $2/3$ of runs satisfy the second condition. 

Secondly, we prove that at least $m'$ of runs satifying the second condition also satisfy the first condition. This concludes the proof of the lemma as $m'$ is independent of $n$ (a precise computation reveals that $m'>0.009$).

Recall that $D_i\leq 1/2^i$ and we assume that $\hl{i}\leq\osmnact i$. Therefore, the probability that \emph{\restart} is not reached during the $i$-th half-life is at least $(1-1/2^i)^{\osmnact i}$ as $t'$ is distributed uniformly and the distance can only get smaller during the half-life. Hence, the probability that in none of the first $n$ half-lives \emph{\restart} is reached is at least $$\prod_{i=1}^n (1-1/2^i)^{\osmnact i}$$ Thus, for every $n$, the probability is greater than $\prod_{i=1}^\infty (1-1/2^i)^{\osmnact i}=:m'$. It remains to show that $m'>0$. This is equivalent to $\sum_{i=1}^\infty \ln(1-1/2^i)^{\osmnact i}>-\infty$, which in turn can be rewritten as 
$$
\osmnact \sum_{i=1}^\infty i\ln\left(\frac{2^i}{2^i-1}\right)
\;\; < \;\;
\infty
$$
Since $\sum_{i=1}^\infty 1/i^2$ converges, it is sufficient to prove that
$$
\ln \left( \frac{2^i}{2^i-1} \right)
\; \in \;
\mathcal O(1/i^3)\,.
$$
We get the result by rewriting the term in the form of an approximation of the derivative of $\ln$ in $2^i-1$ which is smaller than the derivative of $\ln$ in $2^i-1$ because $\ln$ is concave
$$
\ln \left( \frac{2^i}{2^i-1} \right)
\;\; = \;\;
\frac{\ln(2^i)-\ln(2^i-1)}{1}
\;\; \leq \;\;
\ln'(2^i-1) 
\;\;= \;\; 
\frac{1}{2^i-1}
\; \in \;
\mathcal O(1/i^3)\,.
$$\QED

\section{Proofs of Section~\ref{sec:one-results}}
%

In this section, by saying \emph{value} of an event $e$, we mean the fractional part $\fr(\nu(e))$ when the valuation $\nu$ is clear from context.
Furthermore, by $M$ we denote the sum of $u_e$ of all fixed-delay events.

\subsection{Correctness of the region graph construction}
\label{app:region-graph}

The correctness of the region graph construction is based on the fact that configurations in one region can qualitatively reach the same regions in one step.

\begin{lemma}\label{lem:region-graph}
Let $z \region z'$ be configurations and $R$ be a region. We have $\kernel(z,R) > 0$ iff $\kernel(z',R) > 0$.
\end{lemma}
\begin{proof}
For the sake of contradiction, let us fix a region $R$ and a pair of configurations  $z \region z'$ 
such that $\kernel(z,R) > 0$ and $\kernel(z',R) = 0$. Let $z = (s,\nu)$ and $z' = (s,\nu')$.

First, let us deal with the fixed-delay events. Let us assume that the part of $\kernel(z,R)$ contributed by the variable-delay events $V$ is zero, i.e. 
$\sum_{e\in V} \int_0^\infty \hit(\{e\},t) \cdot \win(\{e\},t)\; \de{t} = 0$. Then the set $E$ of fixed-delay events scheduled with the minimal remaining time in $z$ must be non-empty, i.e. some $e \in E$. We have
\begin{align*}
 P(z,R) 
\; = & \;
\occur(s,E)(s') \cdot \indicator{\bar{\nu} \in R} \cdot \prod_{c\in V} \int_{\nu(e)}^\infty f_{c \mid \nu(c)}(y) \; \de{y} > 0 \\
 P(z',R)
\; = & \;
\occur(s,E)(s') \cdot \indicator{\bar{\nu}' \in R} \cdot \prod_{c\in V} \int_{\nu(e)}^\infty f_{c \mid \nu'(c)}(y) \; \de{y} = 0
\end{align*}
where $s'$ is the control state of the region $R$ and $\bar{\nu}$ and $\bar{\nu}'$ are the valuations after the transitions from $z$ and $z'$, respectively.
It is easy to see that from $z \region z'$ we get that
$\bar{\nu} \in R$ iff $\bar{\nu}' \in R$. Hence, $\kernel(z,R)$ and $\kernel(z',R)$ can only differ
in the big product. Let us fix any $c \in V$. 
We show that $\int_{\nu(e)}^\infty f_{c \mid \nu'(c)}(y) \; \de{y}$ 
is positive. Recall that the density function $f_c$ can qualitatively change only on
integral values. Both $z$ and $z'$ have the same order of events' values. Hence, the integral is positive for $\nu'$ iff
it is positive for $\nu$. We get $\kernel(z',R) > 0$ which is a contradiction.

On the other hand, let us assume that there is a
variable-delay event $e\in V$ such that 
$$\int_0^\infty \occur(s,\{e\})(s') \cdot \indicator{\nu_t \in R} \cdot 
f_{e\mid\nu(e)}(t) \cdot \prod_{c\in V \setminus \{e\}} \int_t^\infty
      f_{c \mid \nu(c)}(y) \; \de{y} \; \de{t}
 > 0$$
where $\nu_t$ is the valuation after the transition from $z$ with waiting time $t$.
There must be an interval $I$ such that for every $t \in I$ we have that $f_{e \mid \nu(e)}(t)$ is positive, $\indicator{\nu_t \in R} = 1$, and $\int_t^\infty
      f_{c \mid \nu(c)}(y) \; \de{y} > 0$ for any $c \in V \setminus \{e\}$.
From the definition of the region relation, this interval $I$ corresponds to an interval
between two adjacent events in $\nu$. Since $z \region z'$, there must be also an interval $I'$ such that for every $t \in I'$ we have that
$f_{e \mid \nu'(e)}(t)$ is positive, $\indicator{\nu'_t \in R} = 1$, and $\int_t^\infty
      f_{c \mid \nu'(c)}(y) \; \de{y} > 0$ for any $c \in V \setminus \{e\}$.
Hence, $\kernel(z',R) > 0$, contradiction.\QED
\end{proof}

\subsection{Proof of Proposition~\ref{prop:reach-bscc}}
\label{app:reach}

\begin{reflemma}{lem:delta-sep}
\LemDeltaSep
\end{reflemma} 

\begin{proof}

We divide the $[0,1]$ line segment into $3 \cdot |\events| + 1$ slots of equal length $\delta$. 
Each \emph{value} of a scheduled event lies in some slot. We show how to reach a configuration where the values are separated by empty slots. 

As the time flows, the values shift along the slots. When an event occurs, values of all the newly scheduled events are placed to $0$. The variable-delay events can be easily separated if we guarantee that variable-delay events occur in an interval of time when the first and the last slots of the line segment are empty. 

We let the already scheduled variable-delay events occur arbitrarily.
For each newly scheduled variable-delay event we place a token at the end of an empty slot with its left and right neighbour slots empty as well (i.e. there is no clock's value nor any other token in these three slots). Such slot must always exist since there are more slots that $3 \cdot |\events|$. As the time flows we move the tokens along with the events' values. Whenever a token reaches $1$ on the $[0,1]$ line segment, we do the following. If the valuation of its associated event is not between its lower and upper bound, we move the token to $0$ and wait one more time unit. Otherwise, we let the associated event occur from now up to time $\delta$. Indeed, for any moment in this interval, the first and the last slots of the line segment are empty. The probability that all variable-delay events occur in these prescribed intervals is bounded from below because events' densities are bounded from below. 

The fixed-delay events cause more trouble because they occur at a fixed moment; possibly in an 
occupied slot. If a fixed-delay event always schedules itself (or there is a cycle of fixed-delay events that schedule each other), its value can never be separated from another such fixed-delay event. 
Therefore, we have limited ourselves to at most one ticking event $e$. Observe that every other event has its lifetime -- the length of the chain of fixed-delay events that schedule each other. The lifetime of any fixed-delay event is obviously bounded by $M$ which is the sum of delay of all fixed-delay events in the system. After time $M$, all the old non-ticking events ``die'', all the newly scheduled non-ticking events are separated because they are initially scheduled by a variable-delay event. Therefore, we let the variable-delay events occur as explained above for $m$ steps such that it takes more than $M$ time units in total. We set $m = \lceil M/\delta \rceil$ since each step takes at least $\delta$ time.\QED
\end{proof}

\begin{reflemma}{lem:reach}
\LemReach
Furthermore, $\kernel^k(z,X) > p$ where $X \subseteq R'$ is the set of $(\delta/3^k)$-separated configurations.
\end{reflemma}
\begin{proof}
Let $z \in R_0$, $k \in \Nset$, and $R_0, R_1, \ldots, R_k$ be a path in the region graph to the region $R = R_k$. We can follow this path so that in each step we lose two thirds of the separation. At last, we reach a $(\delta/3^k)$-separated configuration in the target region $R_k$. We get the overall bound on probabilities from bounds on every step. 

In each step either a variable-delay event or a set of fixed-delay events occur. Let $\delta'$ be the separation in the current step.
To follow the region path, a specified event must occur in an interval between two specified values which are $\delta'$-separated. A fixed-delay event occurs in this interval for sure because it has been scheduled this way. For a variable-delay event, we divide this interval into thirds and let the event occur in the middle subinterval. This happens with a probability bounded from below because events' densities are bounded from below. Furthermore, to follow the path in the region graph, no other event can occurs sooner. Every other event has at least $\delta'/3$ to its upper bound; the probability that it does \emph{not} occur is again bounded from below.\QED
\end{proof}
 
\subsection{Proof of Proposition~\ref{prop:bscc-not-sync}}
\label{app:frequencies}

\begin{refproposition}{prop:bscc-not-sync}
\PropBsccNotSync
\end{refproposition}

\begin{proof}

First, we show using the following lemma that $\gssmc$ has a unique invariant measure and that the Strong Law of Large Numbers
holds for $\gssmc$. We prove the lemma later in this subsection.

\begin{reflemma}{lem:small}
 \LemSmall
\end{reflemma}

\noindent
A direct corollary of Lemma~\ref{lem:small} is that the set of configurations is \emph{small}.

\begin{definition}\label{def:small}
  Let $n\in \Nset$, $\varepsilon>0$, and $\smallmeasure$ be a probability measure on $(\configs,\configsfield)$.
  The set $\configs$ is $(n,\varepsilon,\smallmeasure)$-\emph{small} if for all $z\in \configs$ and $A\in \configsfield$
  we have that $\kernel^m(z,A) \geq \varepsilon \cdot \nu(A)$.
\end{definition}

\noindent
Indeed, we can set $n = k + \ell$ and $\varepsilon = \alpha + \beta$ and we get the condition of the definition. 

\begin{corollary}\label{cor:small}
 There is $n \in \Nset$, $\varepsilon > 0$, and $\smallmeasure$ such that $\configs$ is $(n,\varepsilon,\smallmeasure)$-small.
\end{corollary}

From the fact that the whole state space of a Markov chain is small, 
we get the desired statement using standard results on Markov chains on general state space.
We get that $\gssmc$ has a unique invariant measure $\pi$ and that the SLLN holds for $\gssmc$, see~\cite[Theorem 3.6]{BKKKR-HSSC2011}. 

From the SLLN, we directly get that $\dme = E_\pi[\delta]$. Now we show that
$\cme = \frac{E_\pi[\tau]}{E_\pi[W]}$.
\todo{mohl by ses prosim na toto podivat, ze bychom to prodiskutovali?}
Let us consider a run $(s_0,\nu_0)\;(s_1,\nu_1) \cdots$. By $t_i$ we denote 
$\nu_{i+1}(\last)$ -- the time spent in the $i$-th state. We have
\begin{align*}
 \cme(\sigma) &= \lim_{n\rightarrow \infty} \frac{\sum_{i=0}^n \delta(s_i) \cdot t_i}{\sum_{i=0}^n t_i}  
               = \lim_{n\rightarrow \infty} \frac{\sum_{i=0}^n \delta(s_i) \cdot t_i}{n} \cdot \frac{n}{\sum_{i=0}^n t_i}
               = \frac{ \lim_{n\rightarrow \infty} (\sum_{i=0}^n \delta(s_i) \cdot t_i)/n}{\lim_{n\to\infty} (\sum_{i=0}^n t_i)/n} \\
              &= \frac{E_\pi[\tau]}{E_\pi[W]}
\end{align*}
The fact that $\cme(\sigma)$ is well-defined follows from the end which justifies the manipulations with the limits.
It remains to explain the last equality. First, let is divide the space of configurations into a grid $C_\delta$.
Each $\square \in C_\delta$ is a hypercube of configurations of unit length $\delta$. 
By $z_i$, we denote the $i$-th configuration of the run. We obtain
\begin{align*}
 \lim_{n\to\infty} \frac{ \sum_{i=0}^n t_i}{n} 
             &= \lim_{n\to\infty} \sum_{\square \in C_\delta} \frac{ \sum_{i=0}^n \indicator{z_i \in \square} \cdot t_i}{n} \\
             &= \sum_{\square \in C_\delta} \lim_{n\to\infty} \frac{ \sum_{i=0}^n \indicator{z_i \in \square} \cdot t_i}{\sum_{i=0}^n \indicator{z_i \in \square}} \cdot \lim_{n\to\infty} \frac{\sum_{i=0}^n \indicator{z_i \in \square}}{n}
              = (\ast)
\end{align*}
The second limit equals by the SSLN to $\pi(\square)$. 
By taking $\delta \to 0$ we get that $(\ast) = E_\pi[W]$.

By similar arguments we also get that 
\begin{align*}
 \lim_{n\to\infty} \frac{ \sum_{i=0}^n \delta(s_i) \cdot t_i}{n} = E_\pi[\tau]
\end{align*}
\QED
\end{proof}

\noindent
For the proof of Lemma~\ref{lem:small} we introduce 
several definitions
and two auxiliary lemmata. 

\begin{definition}
A path $(s_0,\nu_0) \cdots (s_n,\nu_n)$ is $\delta$-wide if for every $0 \leq i \leq n$ the configuration $(s_i,\nu_i)$ is $\delta$-separated and for every $0 \leq i < n$ any every bounded variable-delay event $e \in \sched(s_i)$ we have $\nu_i(e) + \nu_{i+1}(\last) < u_e - \delta$, i.e. no variable-delay event gets $\delta$ close to its upper bound.

We say that a path $(s_0,\nu_0) \cdots (s_n,\nu_n)$ has a trace $\bar{s}_0 E_1 \bar{s}_1 E_1 \cdots E_n \bar{s}_n$ if $\bar{s}_i = s_i$ for every $0 \leq i \leq n$ and for every $0 < i \leq n$ we can get from $(s_{i-1},\nu_{i-1})$ to $(s_i,\nu_i)$ via
occurrence of the set of events $E_i$ after time $\nu_i(\last)$.

A path $(s_0,\nu_0) \cdots (s_n,\nu_n)$ has a \emph{total time} $t$ if $t = \sum_{i=1}^{n} \nu_i(\last)$. 
%
%

%
%
\end{definition}

The idea is that a $\delta$-wide path can be approximately followed with positive probability. Furthermore, as formalized by the next lemma, if we have different $\delta$-wide paths to the same configuration $z^\ast$ that have the same length and the same trace, we have similar $n$-step behavior (on a set of states specified by some measure $\smallmeasure$). 

\begin{lemma}\label{lem:fuzzying}
 For any $\delta > 0$, any $n \in\Nset$, any configuration $(s_n,\nu_n)$, and any
trace $s_0 E_1 \cdots E_n s_n$ there is a probability measure $\smallmeasure$ and $\beta>0$ such that the following holds. For every $\delta$-wide path $(s_0,\nu_0) \cdots (s_n,\nu_n)$ with trace $T = s_0 E_1 \cdots E_n s_n$ and total time $t \geq M$ and for every $Y \in \configsfield$ we have $\kernel^n((s_0,\nu_0),Y) \geq \beta \cdot \smallmeasure(Y)$.
\end{lemma}

\newcommand{\kusdelty}{4}

\begin{proof}

Recall that $B = \max( \{\ell_e,u_e \mid e \in \events \} \setminus \infty )$.
Notice that the assumptions on the events' densities imply that all delays' densities are bounded by some $\densb>0$ in the following sense. For every $e \in \events$ 
and for all $x\in[0,B]$, $d(x)>\densb$ or equals $0$. Similarly, $\int_{B}^\infty d(x)dx>\densb$ or equals $0$.

We will find a set of configurations $Z$ ``around'' the state $z_n = (s_n,\nu_n)$ and define the
probability measure $\smallmeasure$ on this set $Z$ such that 
$\smallmeasure(Z) = 1$.
Then we show for each measurable $Y \subseteq Z$ the desired property.

%
Intuitively, configurations around $z_n$ are of the form $(s_n,\nu')$ where each $\nu'(e)$ is either exactly $\nu(e)$ or in a small interval around $\nu_n(e)$. 
We now discuss which case applies to which event $e$ for a fixed trace $T$. All the following notions are defined with respect to $T$. We say that the ticking event $g$ is \emph{active until the $i$-th step} if $g \in \sched(s_0) \cap \cdots \cap \sched(s_{i-1})$. We say that an event $e \in \sched(s_n) \cup E_n \cup \{\last\}$ is \emph{originally scheduled in the $i$-th step by $f$} if
\begin{itemize}
 \item either $f=g$ and $g$ is active until the $i$-th step or $f$ is a variable-delay event; and
 \item there is $k \geq 1$ and a chain of events $e_1 \in E_{c_1}, \ldots, e_k \in E_{c_k}$ such that $e_1 = f$, $c_1 = i$, all $e_2,\ldots,e_{k}$ are fixed-delay events, occurence of each $E_{c_i}$ newly schedules $e_{i+1}$, occurence of $E_{c_k}$ newly schedules $e$, and $e \in \sched(s_{c_k}) \cap \cdots \cap \sched(s_{n-2}) \cap (\sched(s_{n-1}) \cup \{\last\})$. 
\end{itemize}

Recall that the special valuation symbol $\last$ denoting the lenght of the last step is also part of the state space. Notice that in the previous definition, we treat $\last$ as an event that is scheduled only in the state $s_{n-1}$. We say that \emph{the last step is variable} if $E_n$ is either a singleton of a variable-delay event or all the events in $E_n$ are originally scheduled by a variable-delay event. Otherwise, we say that \emph{the last step is fixed}.

Intuitively, we cannot alter the value of an event $e$ on the trace $T$ (i.e., $\nu'(e) = \nu(e)$) if the last step is fixed and $e$ is originally scheduled by the ticking event. In all other cases, the value of $e$ can be altered such that $\nu'(e)$ lies in a small interval around $\nu_n(e)$. The rest of the proof is divided in two cases.

\subsubsection{The last step is fixed}

Let us divide the events $e \in \sched(s_n) \cup \{\last\}$ into three sets as follows
\begin{align*}
 e \in A &&& \text{if $e$ is originally scheduled by a variable-delay event and $\fr(\nu_n(e)) \neq 0$;} \\
 e \in B &&& \text{if $e$ is originally scheduled by a variable-delay event and $\fr(\nu_n(e)) = 0$;} \\
 e \in C &&& \text{if $e$ is originally scheduled by the ticking event.}
\end{align*}

Let $a_1, \ldots, a_d$ be the disctinct fractional values of the events $A$ in the valuation $\nu_n$ ordered increasingly by the step in which the corresponding events were originally scheduled. This definition is correct because two events with the same fractional value must be originally scheduled by the same event in the same step. Furthermore, let $F_1,\ldots,F_d$ be the corresponding sets of events, i.e. $\fr(\nu_n(e_i)) = a_i$ for any $e_i\in F_i$.  
We call a configuration $z \region z_n$ such that all events $e \in (B\cup C)$ have the same value in $z$ and $z_n$ a \emph{target} configuration and treat it as a $d$-dimensional vector describing the distinct values for the sets $F_1,\ldots,F_d$. A $\delta$-neighborhood of a target configuration $z$ is
the set of configuration $\{ z + C \mid C \in (-\delta,\delta)^d\}$.
Observe that the $\delta$-neighborhood is a $d$-dimensional space.
We set $Z$ to be the $(\delta/\kusdelty)$-neighborhood of $z_n$ (
the reason for dividing $\delta$ by $\kusdelty$ is technical and will
become clear in the course of this proof).
Let $\smallmeasure_d$ denote the standard
Lebesgue measure on the $d$-dimensional affine space and set 
$\smallmeasure(Y) := \smallmeasure_d(Y) / \smallmeasure_d(Z)$ for any
any measurable $Y \subseteq Z$.

In order to prove the probability bound for any measurable $Y \subseteq Z$,
it suffices to prove it for the generators of $Z$, i.e. for $d$-dimensional
hypercubes centered around some state in $Z$. Let us fix an arbitrary
$z \in Z$ and $\gamma < \delta/\kusdelty$. We set $Y$ to be the 
$\gamma$-neighborhood of $z$. In the rest of the proof we will show
how to reach the set $Y$ from the initial state $(s_0,\nu_0)$ in $n$ steps with
high enough probability.

We show it by altering the original $\delta$-wide path $\sigma = (s_0,\nu_0) \cdots (s_n,\nu_n)$.
Let $t_1,\ldots,t_n$ be the waiting times such that
$t_i = \nu_i(\last)$.
In the first phase, we reach the fixed $z$ instead of the configuration $z_n$. 
We find waiting times $t'_1,\ldots,t'_n$ that induce 
a path $\sigma' = (s_0,\nu_0) \; (s_1,\nu'_1) \; \ldots (s_n,\nu'_n)$ with trace $T$ such that $(s_n,\nu'_n) = z$ and
$t'_i = \nu'_i(\last)$.
In the second phase, we define using $\sigma'$ a set of paths to $Y$. 
We allow for intervals $I_1,\ldots,I_n$ such that
for any choice $\bar{t}_1 \in I_1,\ldots,\bar{t}_n\in I_n$ we get a path 
$\bar{\sigma} = (s_0,\nu_0) \; (s_1,\bar{\nu}_1) \; \ldots (s_n,\bar{\nu}_n)$ such that 
$(s_n,\bar{\nu}_n) \in Y$ and $\bar{t}_i = \bar{\nu}_i(\last)$. 
From the size of the intervals for variable-delay events and from the bound on densities $\densb$ we get 
the overall bound on probabilities. Let us start with the first step.

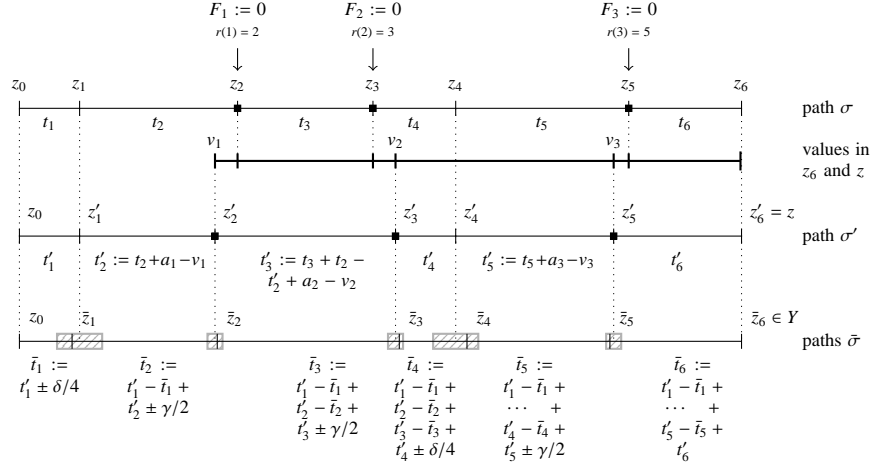
\begin{figure}[t]
\begin{center}
\newcommand{\del}{\delta/4}
\newcommand{\gam}{\gamma/2}

\begin{tikzpicture}[
state/.style={draw,shape=circle,inner sep=0.2em},
axes/.style={-}]


\draw [axes] (0.4,0) -- (10,0);
\foreach \i/\x in {0/0.4, 1/1.2, 2/3.3, 3/5.1, 4/6.2, 5/8.5, 6/10}{
	\draw (\x,-0.1) -- (\x,0.1);
	\node [font=\scriptsize] at (\x,0.3) {$z_\i$};
};

\foreach \i/\x in {1/0.8, 2/2.25, 3/4.2, 4/5.65, 5/7.35, 6/9.25}{
	\node [font=\scriptsize] at (\x,-0.2) {$t_\i$};
};

\foreach \i/\x/\r in {1/3.3/2, 2/5.1/3, 3/8.5/5}{
	\node [font=\scriptsize] at (\x,+1.3) {$F_\i:=0$};
    \node [font=\tiny] at (\x,+1) {$r(\i)=\r$};
    \node [draw,rectangle,fill=black,inner sep=1.2] at (\x,0) {};
    \draw [->] (\x,0.8) -- (\x,0.5);
};


\draw [-|,thick] (3,-0.7) -- (10,-0.7);
\foreach \l/\x/\y in {/3.3/-0.4, v_1/3/-1, /5.1/-0.4, v_2/5.4/-1, /8.5/-0.4, v_3/8.3/-1}{
    \draw [thick] (\x,-.8) -- (\x,-.6);
	\node [font=\scriptsize] at (\x,-0.45) {$\l$};
};

\foreach \x in {3.3, 5.1, 8.5}{
    \draw [dotted] (\x,-.0) -- (\x,-.7);
};


\node [font=\scriptsize,text width=1cm,anchor=west] at (10.7,0) {path $\sigma$};
\node [font=\scriptsize,text width=1cm,anchor=west] at (10.7,-0.7) {values in $z_6$ and $z$};
\node [font=\scriptsize,text width=1cm,anchor=west] at (10.7,-1.7) {path $\sigma'$};
\node [font=\scriptsize,text width=1cm,anchor=west] at (10.7,-3.1) {paths $\bar{\sigma}$};


\draw [axes] (0.4,-1.7) -- (10,-1.7);
\foreach \i/\x in {z_0/0.4, z'_1/1.2, z'_2/3, z'_3/5.4, z'_4/6.2, z'_5/8.3, z'_6=z/10}{
	\draw (\x,-1.6) -- (\x,-1.8);
	\node [font=\scriptsize,anchor=west] at ($(\x,-1.4)+(0,0)$) {$\i$};
};

\foreach \x in {3, 5.4, 8.3}{
    \node [draw,rectangle,fill=black,inner sep=1.2] at (\x,-1.7) {};
};

\foreach \x in {0.4, 1.2, 6.2, 10}{
	\draw [dotted] (\x,-0) -- (\x,-1.8);
};

\foreach \x in {3,5.4,8.3}{
	\draw [dotted] (\x,-0.7) -- (\x,-1.8);
};

\foreach \i/\x/\j in {
1/0.8/,
2/2.15/:=t_2+a_1-v_1,
3/4.3/:=t_3+t_2-t'_2+a_2-v_2, 
4/5.85/,
5/7.3/:=t_5+a_3-v_3, 
6/9.15/}{
	\node [font=\scriptsize,text width=1.5cm,anchor=north,text centered] at (\x,-1.8) {$t'_\i\j$};
};


\foreach \x in {1.2,6.2}{
    \draw [color=black!30,thick]
      ($(\x,-3.1)+(-3mm,-1mm)$) rectangle +(6mm,2mm);
	\pattern [pattern=north east lines,pattern color=black!30]
      ($(\x,-3.1)+(-3mm,-1mm)$) rectangle +(6mm,2mm);
};

\foreach \x in {3, 5.4, 8.3}{
    \draw [color=black!30,thick]
      ($(\x,-3.1)+(-1mm,-1mm)$) rectangle +(2mm,2mm);
	\pattern [pattern=north east lines,pattern color=black!30]
      ($(\x,-3.1)+(-1mm,-1mm)$) rectangle +(2mm,2mm);
};

\draw [axes] (0.4,-3.1) -- (10,-3.1);

\foreach \x/\s in {0.4/z_0, 1.1/\bar{z}_1, 3.03/\bar{z}_2, 5.45/\bar{z}_3, 6.35/\bar{z}_4, 8.25/\bar{z}_5, 10/\bar{z}_6\in Y}{
	\draw (\x,-3.0) -- (\x,-3.2);
    \node [font=\scriptsize,anchor=west] at ($(\x,-2.8)+(0.02,0)$) {$\s$};
};

\foreach \i/\x/\j in {
1/0.8/:=t'_1\pm\del,
2/2.25/:=t'_1-\bar{t}_1+t'_2\pm\gam,
3/4.5/:=t'_1-\bar{t}_1 + t'_2-\bar{t}_2 + t'_3\pm\gam, 
4/5.8/:=t'_1-\bar{t}_1 + t'_2-\bar{t}_2 + t'_3-\bar{t}_3+t'_4\pm\del,
5/7.25/:=t'_1-\bar{t}_1 + \quad \cdots \quad +t'_4-\bar{t}_4+t'_5\pm\gam, 
6/9.35/:=t'_1-\bar{t}_1 + \quad \cdots \quad +t'_5-\bar{t}_5+t'_6\quad}{
	\node [font=\scriptsize,text width=.9cm,anchor=north,text centered] at (\x,-3.2) {$\bar{t}_\i\j$};
};

\foreach \x in {0.4, 1.2, 3, 5.4,6.2,8.3,10}{
	\draw [dotted] (\x,-1.7) -- (\x,-3.1);
};

\end{tikzpicture}
\end{center}
\caption{Illustration of paths leading to the set $Y$. 
The original path $\sigma$ is in the first phase altered to reach the target
state $z$ (its values $v_1$,$v_2$, and $v_3$ are depicted between 
$\sigma$ and $\sigma'$). 
In the second phase, a set of paths that reach $Y$ is constructed
by allowing imprecision in the waiting times -- the transition times are
randomly chosen inside the hatched areas. Notice that
at most $d$ smaller intervals of size $\gamma/2$ can be used to get constant
probability bound with respect to the size of the $d$-dimensional hypercube $Y$. Transitions with fixed-delay are omitted from the illustration (except for the last transition).
%
%
}
\label{fig:fuzzying}
\end{figure}

%
%
%
%
%

Let $v_1,\ldots,v_d$ be the distinct values of the target configuration $z$.
Recall that $|v_i-a_i| < \delta/\kusdelty$ for each $i$.
Let $r(1),\ldots,r(d)$ be the indices such that all events in $F_i$ are originally scheduled in the step $r(i)$. Notice that each $E_{r(i)}$ is a singleton of a variable-delay event.
As illustrated in Figure~\ref{fig:fuzzying}, we set for each $1 \leq i \leq m$
\begin{align*}
 t'_{i} = \begin{cases}
	\ell_e - \nu_{i-1}(e)
	    & \mbox{if } e \in E_i \mbox{ is fixed-delay,}\\
        t_i + \sum_{k=1}^{i-1} (t_k-t'_k) + a_j - v_j
	    & \mbox{if } i = r(j) \mbox{ for } 1 \leq j \leq d \mbox{,} \\
	t_i + \sum_{k=1}^{i-1} (t_k-t'_k)
	    & \mbox{otherwise.}
       \end{cases}
\end{align*}
Intuitively, we adjust the variable-delays in the steps preceding the original scheduling of sets $F_1,\ldots,F_d$ whereas the remaining variable-delay steps are kept in sync with the original path $\sigma$. 
The absolute time of any transition in $\sigma'$ 
(i.e. the position of a line depicting a configuration in Figure~\ref{fig:fuzzying}) 
is not shifted by more than $\delta/4$ since 
$|v_i - a_i| < \delta/\kusdelty$ for any $i$. 
Thus, the difference of any two absolute times is not changed by more than
$\delta/2$. This difference bounds the difference of $|\nu_i(e) - \nu'_i(e)|$ for
any $i$ and $e \in \events$.
Hence, $\sigma'$ is $(\delta/2)$-wide because $\sigma$ is $\delta$-wide.
Furthermore, $\sigma'$ goes through the same regions as $\sigma$ and
performs the same sequence of events scheduling. 
Building on that, the desired property $z'_n = z$ is easy to show.

Next we allow imprecision in the waiting times of $\sigma'$ so that we get a
set of paths of measure linear in $\gamma^{\,d}$. In each step we compensate for the imprecision of the previous 
step. Formally, let $T_i$ denote 
$t'_i + \sum_{k=1}^{i-1} (t'_k - \bar{t}_k)$. For each $1 \leq i \leq m$ we contraint
\begin{align*}
 \bar{t}_{i} \in \begin{cases}
	[T_i,T_i]
	    & \mbox{if } E_i \mbox{ are fixed-delay events,}\\
        (T_i - \frac{\gamma}{2}, \; T_i + \frac{\gamma}{2})
	    & \mbox{if } i = r(j) \mbox{ for } 1 \leq j \leq d\mbox{,} \\
	(T_i - \frac{\delta}{\kusdelty}, \; T_i + \frac{\delta}{\kusdelty})
	    & \mbox{otherwise.}
       \end{cases}
\end{align*}

The difference to $\sigma'$ of any two absolute times is not changed 
by more than $\delta/2$ because the imprecision of any step 
is bounded by $\delta/\kusdelty$. Because $\sigma'$ is $(\delta/2)$-wide, any path
$\bar{\sigma}$ goes through the same 
regions as $\sigma'$. The difference of the value of events in any $F_i$ in 
the state $\bar{z}_n$ from the state $z$ is at most $\gamma/2$ because 
it is only influenced by the imprecision of the step preceding its original scheduling. Hence, $\bar{z}_n \in Y$.

By $v$ we denote the number of variable-delay singletons among $E_1,\ldots,E_n$. From the definition of $\kernel$, it is easy to prove by that
\begin{align*}
\kernel^n(z_0,Y) 
 \quad \geq \quad
\pmin^n \cdot (\densb \cdot \gamma)^d \cdot  (\densb \cdot \delta/2)^{v-d}
 \quad \geq \quad
(\pmin \cdot \densb / 2)^n \cdot \gamma^d \cdot \delta^{n-d}
\end{align*}

Since $\smallmeasure_d(Y) = (2 \cdot \gamma)^d$ 
and $\smallmeasure_d(Z) = (2\cdot \delta/\kusdelty)^d$, we have
$\smallmeasure(Y) = \smallmeasure_d(Y) / \smallmeasure_d(Z) =
(4\gamma/\delta)^d$. We get
$\kernel^n(z_0,Y) \; \geq \; \smallmeasure(Y) \cdot 
(\delta \cdot \pmin \cdot \densb / 8)^n$ and conclude the proof of this case
by setting $\varepsilon = (\delta \cdot \pmin \cdot \densb / 8)^n$.

\subsubsection{The last step is variable}

The rest of the proof proceeds in a similar fashion as previously, we reuse the same notions and the same notation. We only redefine the differences: the neighbourhood and the way the paths are altered.  

We call $(s,\nu) \region z_n$ a \emph{target} configuration if there is $y \in \Rset$ such that for all events $e \in C$ we have $\nu(e) - \nu_n(e) = y$ and for all events $e \in B$ we have $\nu(e) = \nu_n(e)$. We set $g=d+1$ if $C$ is non-empty, and $g=d$, otherwise. We treat a target configuration as a $g$-dimensional vector describing the distinct values for the sets $F_1,\ldots,F_d$ and the value $y$, if necessary. Again, a $\delta$-neighborhood of a target configuration $z$ is
the set of configuration $\{ z + C \mid C \in (-\delta,\delta)^{g}\}$.
We set $Z$ to be the $(\delta/\kusdelty)$-neighborhood of $z_n$ and set 
$\smallmeasure(Y) := \smallmeasure_{g}(Y) / \smallmeasure_{g}(Z)$ for any
any measurable $Y \subseteq Z$. We fix $Y$ to be a $\gamma$-neighborhood of a fixed $z \in Z$.

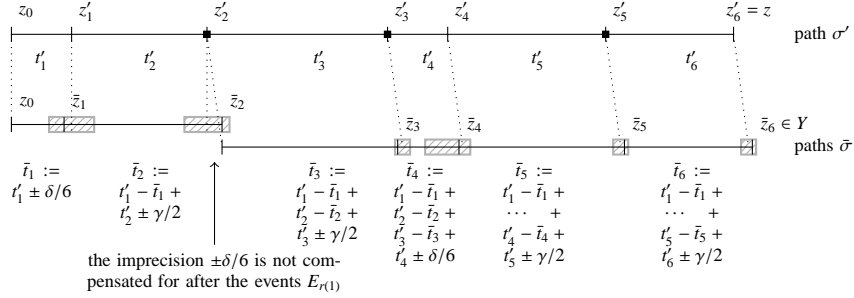
\begin{figure}[t]
\begin{center}
\newcommand{\del}{\delta/6}
\newcommand{\gam}{\gamma/2}

\begin{tikzpicture}[
state/.style={draw,shape=circle,inner sep=0.2em},
axes/.style={-}]


\node [font=\scriptsize,text width=1cm,anchor=west] at (10.7,-1.7) {path $\sigma'$};
\node [font=\scriptsize,text width=1cm,anchor=west] at (10.7,-3.2) {paths $\bar{\sigma}$};

\draw [->] (3.1,-4.5) -- (3.1,-3.4);
\node [font=\scriptsize,text width=5cm,anchor=north,text centered] at (3.1,-4.5) {
the imprecision $\pm\del$ is not compensated for after the events $E_{r(1)}$};


\draw [axes] (0.4,-1.7) -- (10,-1.7);
\foreach \i/\x in {z_0/0.4, z'_1/1.2, z'_2/3, z'_3/5.4, z'_4/6.2, z'_5/8.3, z'_6=z/10}{
	\draw (\x,-1.6) -- (\x,-1.8);
	\node [font=\scriptsize] at ($(\x,-1.4)+(0.2,0)$) {$\i$};
};

\foreach \x in {3, 5.4, 8.3}{
    \node [draw,rectangle,fill=black,inner sep=1.2] at (\x,-1.7) {};
};

\foreach \i/\x/\j in {
1/0.8/,
2/2.25/,
3/4.5/, 
4/5.95/,
5/7.4/, 
6/9.45/}{
	\node [font=\scriptsize,text width=1.5cm,anchor=north,text centered] at (\x,-1.8) {$t'_\i\j$};
};


\foreach \x in {1.2, 3}{
    \draw [color=black!30,thick]
      ($(\x,-2.9)+(-3mm,-1mm)$) rectangle +(6mm,2mm);
	\pattern [pattern=north east lines,pattern color=black!30]
      ($(\x,-2.9)+(-3mm,-1mm)$) rectangle +(6mm,2mm);
};

\foreach \x in {6.2}{
    \draw [color=black!30,thick]
      ($(\x,-3.2)+(-3mm,-1mm)$) rectangle +(6mm,2mm);
	\pattern [pattern=north east lines,pattern color=black!30]
      ($(\x,-3.2)+(-3mm,-1mm)$) rectangle +(6mm,2mm);
};

\foreach \x in {5.6, 8.5, 10.2}{
    \draw [color=black!30,thick]
      ($(\x,-3.2)+(-1mm,-1mm)$) rectangle +(2mm,2mm);
	\pattern [pattern=north east lines,pattern color=black!30]
      ($(\x,-3.2)+(-1mm,-1mm)$) rectangle +(2mm,2mm);
};

\draw [axes] (0.4,-2.9) -- (3.2,-2.9);
\draw [axes] (3.2,-3.2) -- (10.25,-3.2);

\foreach \x/\s in {0.4/z_0, 1.1/\bar{z}_1,3.2/\bar{z}_2}{
	\draw (\x,-2.8) -- (\x,-3.0);
    \node [font=\scriptsize,anchor=west] at (\x,-2.6) {$\s$};
};
\foreach \x/\s in {3.2/, 5.535/\bar{z}_3, 6.35/\bar{z}_4, 8.55/\bar{z}_5, 10.25/\bar{z}_6\in Y}{
	\draw (\x,-3.1) -- (\x,-3.3);
    \node [font=\scriptsize,anchor=west] at (\x,-2.9) {$\s$};
};

\foreach \i/\x/\j in {
1/0.8/:=t'_1\pm\del,
2/2.25/:=t'_1-\bar{t}_1+t'_2\pm\gam,
3/4.6/:=t'_1-\bar{t}_1 + t'_2-\bar{t}_2 + t'_3\pm\gam, 
4/5.9/:=t'_1-\bar{t}_1 + t'_2-\bar{t}_2 + t'_3-\bar{t}_3+t'_4\pm\del,
5/7.35/:=t'_1-\bar{t}_1 + \quad \cdots \quad +t'_4-\bar{t}_4+t'_5\pm\gam, 
6/9.45/:=t'_1-\bar{t}_1 + \quad \cdots \quad +t'_5-\bar{t}_5+t'_6\pm\gam}{
	\node [font=\scriptsize,text width=.9cm,anchor=north,text centered] at (\x,-3.3) {$\bar{t}_\i\j$};
};

\foreach \x in {0.4, 1.2, 3}{
	\draw [dotted] (\x,-1.7) -- (\x,-3);
};

\foreach \x in {3, 5.4, 6.2, 8.3, 10}{
    \draw [dotted] (\x,-1.7) -- ($(\x,-3.2)+(0.2,0)$);
};

\end{tikzpicture}
\end{center}
\caption{Illustration of construction of $\bar{\sigma}$ for the empty set $C$ and the last step variable. 
%
%
}
\label{fig:fuzzying2}
\end{figure}

%
%
%
%
%
The path $\sigma'$ is obtained from the $\sigma$ in the same way as before. We need to allow imprecision in the waiting times of $\sigma'$ so that we get a
set of paths of measure linear in $\gamma^{\,g}$. 
\begin{itemize}
 \item For the case $g = d+1$ it is straightforward as we make the last step also with imprecision $\pm \gamma/2$. Precisely

\begin{align*}
 \bar{t}_{i} \in \begin{cases}
	[T_i,T_i]
	    & \mbox{if } E_i \mbox{ are fixed-delay events,}\\
        (T_i - \frac{\gamma}{2}, \; T_i + \frac{\gamma}{2})
	    & \mbox{if } i = r(j) \mbox{ for } 1 \leq j \leq d\mbox{ or $i=m$,} \\
	(T_i - \frac{\delta}{\kusdelty}, \; T_i + \frac{\delta}{\kusdelty})
	    & \mbox{otherwise}
       \end{cases}
\end{align*}
where $E_n$ are originally scheduled in the $m$-th step if $E_n$ are fixed-delay events, and $m$ equals $n$, otherwise. The difference of the value of events in any $F_i$ in 
the state $\bar{z}_n$ from the state $z$ is at most $\gamma$ because 
it is influenced by the imprecision of the step preceding its original scheduling and also by the imprecision of the last step. Events in $C$ have the difference of the value at most $\gamma/2$ because of the last step. Hence, $\bar{z}_n \in Y$.
Again, we get that $\kernel^n(z_0,Y) \; \geq \; \smallmeasure(Y) \cdot 
(\delta \cdot \pmin \cdot \densb / 8)^n$ and conclude the proof
by setting $\varepsilon = (\delta \cdot \pmin \cdot \densb / 8)^n$.

 \item For the case $g = d$ it is somewhat tricky since
only at most $d$ choices of waiting times can have their precision dependent 
on $\gamma$. In each step we compensate for the imprecision of the previous 
step. Only the imprecision of the step preceding the first scheduling $E_1$ is not
compensated for. Otherwise, it would influence the value of events $E_1$ in
$\bar{z}_n$. Let $T^a_i$ denote $t'_i + \sum_{k=a}^{i-1} (t'_k - \bar{t}_k)$. 
As illustrated in Figure~\ref{fig:fuzzying2}, we contraint
\begin{align*}
 \bar{t}_{i} \in \begin{cases}
	[T_i^1,T_i^1]
	    & \mbox{if } E_i \mbox{ are fixed-delay events,}\\
	(T_i^1 - \frac{\delta}{6}, \; T^1_i + \frac{\delta}{6})
	    & \mbox{if } i \leq r(1) \mbox{,} \\
        (T^{r(1)+1}_i - \frac{\gamma}{2}, \; T^{r(1)+1}_i + \frac{\gamma}{2})
	    & \mbox{if } i = r(j) \mbox{ for } 2 \leq j \leq d  \mbox{ or } i=m \mbox{,} \\
	(T^{r(1)+1}_i - \frac{\delta}{6}, \; T^{r(1)+1}_i + \frac{\delta}{6})
	    & \mbox{otherwise.}
       \end{cases}
\end{align*}

The difference to $\sigma'$ of any two absolute times is not changed 
by more than $3\cdot \delta/6 = \delta/2$ because the imprecision of any step 
is bounded by $\delta/6$. Because $\sigma'$ is $(\delta/2)$-wide, any path
$\bar{\sigma}$ goes through the same 
regions as $\sigma'$. The difference of the value of events $E_1$ in 
the state $\bar{z}_n$ from the state $z$ is at most $\gamma/2$ because 
it is only influenced by the imprecision of the last step. The difference of 
any other event $e$ is at most $2 \cdot \gamma/2$ because 
it is influenced by the imprecision of the step preceding the original scheduling of $e$, 
as well. Hence, $\bar{z}_n \in Y$.

\noindent
Now, we get that $\kernel^n(z_0,Y) \; \geq \; \smallmeasure(Y) \cdot 
(\delta \cdot \pmin \cdot \densb / 12)^n$ and conclude the proof
by setting $\varepsilon = (\delta \cdot \pmin \cdot \densb / 12)^n$. \QED
\end{itemize}

\end{proof}

\begin{lemma}\label{lem:wide-path}
 Let $\delta > 0$ and $R$ be a region such that the ticking event $e$ is either not scheduled or has the greatest value among all events scheduled in $R$. There is $n \in \Nset$, $\delta'>0$, a configuration $z^\ast$, and a trace $s_0 E_1 \cdots E_n s_n$ such that 
from any $\delta$-separated $z \in R$, there is a $\delta'$-wide path to $z^\ast$ with trace $s_0 E_1 \cdots E_n s_n$ and total time $t \geq M$.
\end{lemma} 
\begin{proof}
We use a similar concept as in the proof of Lemma~\ref{lem:delta-sep}. Let us fix a $\delta$-separated $z \in R$. Let $a$ be the greatest value  of all event scheduled in $z$. Observe, that no value is in the interval $(a,a+\delta)$. When we build the $\delta$-wide path step by step, we use a variable $s$ denoting start of this interval of interest which flows with time. Before the first step, we have $s := a$. After each step, which takes $t$ time, we set $s := \fr(s+t)$.

In the interval $[s,s+\delta]$ we make a grid of $3 \cdot |\events| + 1$ points that we shift along with $s$, 
and set $\delta' = \delta / (3 \cdot |\events| + 1)$. On this grid, a procedure similar to the $\delta$-separation takes place.
We 
build the $\delta'$-wide path by choosing sets of events $E_i$ to occur, waiting times $t_i$ of the individual transitions, and target states $z_i$ after each transition so that 
\begin{itemize}
 \item every variable-delay event occurs exactly at an empty point of the grid (i.e. at a time when an empty point has value $0$), and 
 \item the built path is ``feasible'', i.e. all the specified events can occur after the specified waiting time, and upon each occurrence of a specified event we move to the specified target state with positive probability,
\end{itemize}
These rules guarantee that the path we create is $\delta'$-wide. Indeed, the initial configuration is $\delta$-separated for $\delta > \delta'$,  upon every new transition, the $\delta'$-neighborhood of $0$ is empty, and every variable-delay event occurs at a point different form its current point, whence it occurs at least $\delta'$ prior to its upper bound. It is easy to see that such 
choices are possible
since there are only $\events$ events, but $3 \cdot |\events| + 1$ points.

Now we show that this procedure lasts only a fixed amount of steps before all the scheduled events lie on the grid. Notice that if the ticking event is scheduled in $R$, it lies at a point of the grid from the very beginning because we define the grid adjacent to its value. If it is not scheduled, it can get scheduled only by a variable-delay event which occurs already at a point of the grid. Values of any other scheduled fixed-delay event gets eventually placed at a point of a grid. Indeed, every such event gets scheduled by a variable-delay event next time, since we assume a \nice GSMP. We now that after time $M$, all the non-ticking fixed-delay events are either not scheduled or lie on the grid. Each step takes at least $\delta'$ time. In total, after $n = \lceil M/\delta' \rceil + 1$ steps with trace $E_1,\ldots,E_n$, we can set $z^\ast := z_n$.

It remains to show that from any other $\delta'$-separated configuration $z'\in R$, we can build a $\delta'$-wide path of length $n$, 
with trace $E_1,\ldots,E_n$ that ends in $z^\ast$. We start in the same region. 
From the definition of the region relation and from the fact that all events occur in the empty interval $(a,a+\delta)$ we get the following.
By appropriately adjusting the waiting times so that the events occur at the same points of the grid as before, we can follow the same trace and the same control states (going through the same regions) and build a path $z'_0 \ldots z'_n$ such that $z'_n = z_n$. 
Indeed, all scheduled events have the same value in $z'_n$ as in $z_n$ because 
they lie on the same points of the grid. In fact, this holds for $z'_{n-1}$ and $z_{n-1}$ as well (because the first $n-1$ steps take more than $M$ time units) except for the value of $\last$. Finally, also $\last$ has the same value in $z'_n$ as in $z_n$ because there is no need to alter the waiting time in the last step.
By the same arguments as before, the built path is also $\delta'$-wide. \QED

\end{proof}

\begin{reflemma}{lem:small}
 \LemSmall
\end{reflemma}

\begin{proof}
We choose some reachable region $R$ such that the ticking event $e$ is either not scheduled in $R$ or $e$ has the greatest fractional part among all the scheduled events. There clearly is such a region. We fix a sufficiently small $\delta > 0$ and choose $C$ to be the set of $\delta$-separated configurations in $R$. Now, we show how we fix this $\delta$.

It is a standard result from the theory of Markov chains, see e.g.~\cite[Lemma 8.3.9]{Rosenthal:book}, 
that in every ergodic Markov chain there is $n$ such that between
any two states there is a path of length exactly $n$. 
The same result holds for the aperiodic region graph $\regiongraph$. From Lemma~\ref{lem:delta-sep}, we reach in $m$ steps from any $z \in \configs$ a $\delta'$-separated configuration $z'$ with probability at least $q$. From $z'$, we have a path of length $n$ to the region $R$. From Lemma~\ref{lem:reach}, we have $p > 0$ such that we reach $R$ from $z'$ in $n$ steps with probability at least $p$.
Furthermore, we end up in a $(\delta'/3^n)$-separated configuration of the region $R$. Hence, we set $\delta = \delta'/3^n$ and obtain the first part of the lemma.

The second part of the lemma is directly by connecting Lemmata~\ref{lem:wide-path} and~\ref{lem:fuzzying}.
\end{proof}

\section{Proof of Theorem~\ref{thm:approx}}
\label{sec:app-results}

It only remains to prove the two Claims. 

\begin{claim}
\claimApproxReach
\end{claim}
\begin{proof}
Let $Y$ denote the union of regions from which the BSCC $\mathcal{B}$ is reachable.
By Lemmata~\ref{lem:delta-sep} and~\ref{lem:reach} we have $p,q > 0$ and $m \in \Nset$ and $k < \sizereg$ such that from any $z \in Y$ we reach $X$ in $m+k$ steps with probability at least $p \cdot q$. We get the first part by setting $n_1 = m+k$ and $p_1 = p \cdot q$. Indeed, if the process stays in $Y$ after $n_1$ steps, it has the same chance to reach $X$ again, if the process reaches $X$, it never leaves it, and if the process reaches $\configs \setminus (X \cup Y)$, it has no chance to reach $X$ any more.

By Corollary~\ref{cor:small} in Appendix~\ref{app:frequencies}, $\configs$ is $(n,\varepsilon,\smallmeasure)$-small. By Theorem~8 of \cite{RR:GSSMC-PS} we thus obtain that for all $x\in\configs$ 
and all $i\in \Nset$,
    \[
    \sup_{A\in \configsfield} |\kernel^i(x,A)-\pi(A)|\quad \leq \quad
    (1-\varepsilon)^{\lfloor i/n\rfloor}
    \]
which yields the second part by setting $A = \{(s,\nu)\in\configs\mid s=\measured\}$ and observing $d = \pi(A)$ and $A \in \configsfield$
\QED
\end{proof}

\begin{claim}
\claimFunctionW
\end{claim}
\begin{proof}
Let $(s,\nu)$ be a configuration, $C$ and $D$ the set of variable-delay and fixed-delay events scheduled in $s$, respectively. For nonempty $D$ let $T=\min_{d\in D}(\ell_d-\nu(d))$ be the time the first fixed-delay event can occur; for $D=\emptyset$ we set $T=\infty$. 
%
The probability that the transition from $(s,\nu)$ occurs within time $t$ is 
$$F(t)=\begin{cases}
         1- \prod_{c\in C}\int_t^\infty f_{c \mid \nu(c)}(x)\;\de{x} & \text{for }0<t<T,\\
         1 & \text{for }t\geq T
       \end{cases}
$$
as non-occurrences of variable-delay events are mutually independent.
Observe that $F(t)$ is piece-wise differentiable on the interval $(0,T)$, we denote by $f(t)$ its piece-wise derivative. The expected waiting time in $(s,\nu)$ is 
\begin{equation}\label{eq:W}
   W((s,\nu))=\begin{cases}
              \int_0^T t\cdot f(t)\;\de{t} + T\cdot (1-F(T))&\text{for }T<\infty,\\
              \int_0^\infty t\cdot f(t)\;\de{t}  & \text{for }T=\infty.
             \end{cases}
  \end{equation} 
Recall that for each variable-delay event $e$, the density $f_e$ is continuous and bounded as it is defined on a closed interval.
Therefore, $f_{e \mid t}$ are also continuous, hence $F$ and $f$ are also continuous with respect to $\nu$ and with respect to $t$ on $(0,T)$. Thus $W$ is continuous for $T$ both finite and infinite. 
Moreover, for finite $T$, $W$ is bounded by $T$ which is for any $(s,\nu)$ smaller than $\max_{d\in\fixed}\ell_d$. Hence, $E_\pi[W]$ is finite. 
%
For $T=\infty$, $E_\pi[W]$ is finite due to the assumption that each $f_e$ has finite expected value.
%
%
%
%
%
 \QED
\end{proof}







\end{document}